\documentclass[10pt]{article}
\usepackage[usenames]{color} 
\usepackage{amssymb} 
\usepackage{amsmath} 
\usepackage{amsthm}
\usepackage{bm}
\usepackage{framed} 
\usepackage[english]{babel}
\usepackage[utf8]{inputenc} 
\usepackage{hyperref}
\usepackage{parskip}
\usepackage{eulervm}
\usepackage{charter}
\usepackage{natbib}

\newtheorem{theo}{Theorem}[section]
\newtheorem{cor}[theo]{Corollary}
\newtheorem{lem}[theo]{Lemma}
\newtheorem{pro}[theo]{Proposition}
\theoremstyle{definition}
\newtheorem{defi}[theo]{Definition}
\theoremstyle{remark}
\newtheorem{rem}[theo]{Remark}

\title{The geometric Laplace transform: Definition, existence and properties of the Geometric Algebra Laplace transform}
\author{Manel Velasco$^1$, Arnau Dòria-Cerezo$^2$, and Isiah Zaplana$^2$}
\date{\begin{flushleft}
\small{$^1$ Automatic Control Department, Universitat Politècnica de Catalunya, Barcelona, Spain\\
      $^2$ Inst. of Industrial and Control Engineering, Universitat Politècnica de Catalunya, Barcelona, Spain}
\end{flushleft}}
\begin{document}
	\maketitle
	\tableofcontents
	
	\setlength{\parindent}{0pt}
	\section{Introduction}\label{Sec:Intro}
	
	Recent publications \citep{Vel23,DVZM2024} have started to explore the application of Geometric Algebra (GA) to the modeling, analysis and control of dynamical systems and, in particular, electrical circuits. Since a crucial element there is to transform the ordinary differential equations governing the dynamical system which models the systems' behavior from the real domain to the Laplace domain, a definition of the Laplace transform in GA is needed.
	
	There is no much literature regarding the Laplace transform in hipercomplex algebras, i.e., algebras extending the complex numbers such as quaternions, octonions, geometric and Cayley-Dickson algebras. Some of the first works in the subject are the series of two papers by Sommen \citep{Sommen82,Sommen83} which focused on the Laplace and Fourier transforms for multivector-valued functions of a multivector variable. His work constitutes one of the base pillars of the field of hipercomplex and Clifford analysis. Later on, and on a more applied side, \cite{Ell92} introduced the following Laplace transform for quaternion-valued functions of two real variables $h(t,\tau)$: 
	\begin{equation}
		\mathcal{L}\{h\}(s,z)=\int_{0}^{\infty}\int_{0}^{\infty}e^{-st}h(t,\tau)e^{-z\tau}dtd\tau.
	\end{equation}
	This Laplace transform is defined for functions of two real variables and, therefore, requires two exponential functions, $ e^{-st} $ and $ e^{-z\tau} $, to compute the transformation. In addition, it can only be applied to quaternion-valued functions of a real variable.
	
	More recent works include the definition of the Fourier transform for quaternions \citep{Bujack2013,Hitzer16}, in geometric algebra \citep{Hitzer07_art,Hitzer07}, and the Laplace transform for quaternions, octonions, and Cayley-Dickson algebras \citep{Ludkovsky08,Ludkovski12}. These last two works address the Laplace transform for non-associative hypercomplex algebras, with the latter taking advantage of the proposed general formulation for the quaternion case. Finally, \cite{Cai2017} and \cite{Bau2023} propose definitions of the Laplace transform for quaternion-valued functions of a real variable (in the former) and a complex variable (in the latter). They also enumerate and prove all the classical properties of the newly proposed transforms. 
	
	In the present work, we extend these ideas by introducing a definition of the Laplace transform within the framework of Geometric Algebra (GA). In particular, our definition and its properties are applicable to geometric algebras $\mathcal{G}_{q,r}$ where $q+r\leq 5$. However, to simplify notation and computations, we will primarily use $\mathcal{G}_2$ throughout this manuscript, with occasional remarks on how to extend the approach to higher-dimensional geometric algebras. This version of the Laplace transform, which we refer to as the geometric Laplace transform, is analogous to the classical Laplace transform because, in the Laplace domain, the transformed functions correspond to those obtained with the classical Laplace transform, with the solely difference that they are functions of a multivector variable, as demonstrated in Section \ref{Sec:Ex}. We also enumerate and prove the main properties of the geometric Laplace transform. This new definition provides a formal framework for current and future developments in the modeling, analysis, and control of dynamical systems, as well as in other applications.
	
	Furthermore, the approach presented in this work differs from that introduced in \citep{Cai2017,Bau2023} in several ways. First, the Laplace transform defined there is applicable only to quaternion-valued functions (either of a real or complex variable), while the geometric Laplace transform is applicable to multivector-valued functions of a real variable over any geometric algebra $ \mathcal{G}_{q,r} $ such that $q+r\leq5$. In addition, as will be shown in Sections \ref{Sec:GLT} and \ref{Sec:FP}, the transformed function is also a multivector-valued function of a multivector variable, which is fundamental for certain engineering applications (such as the one triggering this work \citep{Vel23}), while in the other two works, it is a quaternion-valued function of a complex variable. Finally, some additional properties based on the involution operators of GA have also been enumerated and proved.
	
	The manuscript is organized as follows. Section \ref{Sec:IntroGA} introduces the basics of GA. In Section \ref{Sec:GLT}, we present the definition of the geometric Laplace transform. The main properties of this transform are then enumerated and proved in Section \ref{Sec:FP}. Section \ref{Sec:Ex} provides the geometric Laplace transform of some well-known functions, and finally, conclusions are discussed in Section \ref{Sec:Conc}.
	
	\section{Introduction to Geometric Algebra}\label{Sec:IntroGA}
	
	To facilitate the reading of this paper, a brief introduction to Geometric Algebra (GA) is provided in this section. Readers interested in a more detailed treatment of the subject are referred to the classical texts \cite{Hes84,Hes01,Doran03}.
	
	Let $\mathbb{R}^{q,r}$ be a the pseudo-Euclidean space with orthonormal basis $\{e_{1},\dots,e_{q},$ $e_{q+1},\dots,e_{q+r}\}$, where each basis element $e_i$ square to either $1$ or $-1$, i.e., $e_i^2 = e_i\cdot e_i = 1$ if $1\leq i\leq q$ and $e_i^2 = -1$ if $q+1\leq j\leq r$. The geometric algebra of $\mathbb{R}^{q,r}$, denoted by $\mathcal{G}_{q,r}$, is a vector space where the operations defined in $\mathbb{R}^{q,r}$, i.e., the addition and multiplication by scalars, are extended naturally. An additional operation, the geometric product, is defined, and acts on vectors of the algebra as follows: 
	\begin{equation}\label{geometricproduct}
		v_{1}v_{2} = v_{1}\cdot v_{2} + v_{1}\wedge v_{2},\,\text{ for } v_{1},v_{2}\in\mathbb{R}^{q,r},
	\end{equation}
	where $\cdot$ denotes the inner or dot product and $\wedge$ denotes the outer product. 
	
	The outer product of two vectors $v_{1},v_{2}$ is a new element of $\mathcal{G}_{q,r}$, which is termed a bivector, is said to have grade two and is denoted by $v_{1}\wedge v_{2}$. By extension, the outer product of a bivector with a vector is known as a trivector and is denoted by $(v_{1}\wedge v_{2})\wedge v_{3}$. Clearly, trivectors have grade three. This can be generalized to an arbitrary dimension. Thus,
	\begin{equation}\label{rblade}
		(v_{1}\wedge v_{2} \wedge \dots \wedge v_{k-1})\wedge v_{k}
	\end{equation}
	denotes an \textit{$k$-blade}, an element of $\mathcal{G}_{q,r}$ with grade $k$. 
	
	A bivector $v_{1}\wedge v_{2}$ can be interpreted as the oriented area defined by the vectors $v_{1}$ and $v_{2}$. Thus, $v_{2}\wedge v_{1}$ has opposite orientation and, from that, the anticommutativity of the outer product can be deduced. Analogously, a trivector is interpreted as the oriented volume defined by its three composing vectors. Since the volume generated by $(v_{1}\wedge v_{2})\wedge v_{3}$ is the same as the volume generated by $v_{1}\wedge (v_{2}\wedge v_{3})$, it is also deduced that the outer product is associative. Therefore, $k$-blades can be denoted simply as:
	\begin{equation}\label{ibladebis}
		v_{1}\wedge v_{2} \wedge \dots \wedge v_{k}.
	\end{equation}
	Linear combinations of $k$-blades are known as $k$-vectors, while linear combinations of $k$-vectors, for $0\leq k\leq q+r$, are called multivectors. Multivectors are the most important elements of geometric algebra.
	
	Applied to the basis elements $\{e_i\}$, the geometric product acts as follows:
	\begin{equation}\label{elementsofbasis}
		e_{i}e_{j}  = \left\{\begin{array}{lcl}
			1 & \text{for} & i=j\text{ and }i\leq q\\
			-1 & \text{for} & i=j\text{ and } q< i\leq q+r\\
			e_{i}\wedge e_{j} & \text{for} & i\neq j
		\end{array}\right.
	\end{equation} 
	Then, $\{e_{1},\dots,e_q,e_{q+1},\dots e_{q+r}\}$ can be expanded to a basis of $\mathcal{G}_{q,r}$ that contains, for each $0\leq k\leq q+r$, $C(q+r,k)$ grade $k$ elements:
	\begin{equation}\label{basisGn}
		\begin{split}
			&\text{Scalar: } 1\\
			&\text{Vectors: } e_{1},\dots,e_{q+r}\\
			&\text{Bivectors: } \{e_{i}\wedge e_{j}\}_{1\leq i<j\leq q+r}\\
			&\text{Trivectors: } \{e_{i}\wedge e_{j}\wedge e_{k}\}_{1\leq i<j<k\leq q+r}\\
			&\vdots\\
			&k-\text{vectores: } \{e_{i_{1}}\wedge\dots\wedge e_{i_{k}}\}_{1\leq i_{1}<\dots<i_{k}\leq q+r}\\
			&\vdots\\
			&\text{$q+r$-blade: } e_{1}\wedge\dots\wedge e_{q+r}
		\end{split}
	\end{equation}
	which sums up to total of $2^{q+r}$ elements. Understanding how the geometric product acts on the basis elements of $\mathcal{G}_{q,r}$ allows for its extension to arbitrary multivectors. The grade $q+r$ element $e_{1}\wedge\dots\wedge e_{q+r}$ is known as the pseudoscalar and is usually denoted by $I$. Pseudoscalars allow for the definition of one of the main operators of geometric algebra, the dual operator. Its action over an $k$-vector $A_{k}$ is:
	\begin{equation}\label{dual}
		A_{k}^{\ast} = A_kI,
	\end{equation}
	where $A_{k}^{\ast}$ is an $(q+r-k)$-vector. In particular, for two multivectors $A$ and $B$, the following identity holds:
	\begin{equation}\label{deMorgan}
		(A\wedge B)^\ast = A\cdot B^\ast.
	\end{equation}
	 An important class of linear operators in $\mathcal{G}_{q,r}$ are the grade-$k$ projection operators. They are denoted by $\langle\cdot\rangle_{k}$ for $0\leq k\leq q+r$. When applied to an arbitrary multivector $A$, $\langle A\rangle_{k}$ projects onto the grade-$k$ components in $A$, i.e., it returns the components of $A$ that can be expressed as a linear combination of $\{e_{i_{1}}\wedge\dots\wedge e_{i_{k}}\}_{1\leq i_{1}<\dots<i_{k}\leq q+r}$. The elements $\langle A\rangle_{0}$ are said to be the scalar part of $A$. Clearly, if $A_{k}$ denotes a $k$-vector, then $\langle A_{k}\rangle_{k} = A_{k}$.
	
	Using these operators, general multivectors $A\in\mathcal{G}_{q,r}$ can be expressed as:
	\begin{equation}\label{multivectors}
		A = \langle A\rangle_{0} + \langle A\rangle_{1} + \dots + \langle A\rangle_{q+r}.
	\end{equation}
	Hence, the set of all $k$-vectors for a given $1\leq k\leq q+r$ is a vector subspace of $\mathcal{G}_{q,r}$ denoted by $\langle\mathcal{G}_{q,r}\rangle_{k}$ and spanned by $B_{k} = \{e_{i_{1}}\wedge\dots\wedge e_{i_{k}}\}_{1\leq i_{1}<\dots<i_{k}\leq q+r}$.
	
	The multivector representation (\ref{multivectors}) is very useful in defining another important operator in $\mathcal{G}_{q,r}$. This linear operator is known as the \textit{reversion involution operator} and is denoted by the superscript $\sim$. The reversion is defined over the geometric product of $n$ vectors as:
	\begin{equation}\label{reverse}
		(a_{1}\cdots a_{n})^{\sim} = a_{n}\cdots a_{1}.
	\end{equation}
	Applied to $k$-vectors:
	\begin{equation}\label{reverseblades}
		\widetilde{A}_k = (-1)^{\frac{k(k-1)}{2}}A_{k}
	\end{equation}
	due to the anticommutativity of the outer product. Finally, since reversion is a linear operator, the reverse of an arbitrary multivector is:
	\begin{equation}\label{reversemultivector}
		\begin{split}
			\widetilde{A} &= \langle\widetilde{A}\rangle_{0} + \dots + \langle\widetilde{A}\rangle_{q+r}\\
			&= \langle A\rangle^{\sim}_{0} + \dots + \langle A\rangle^{\sim}_{q+r}\\
			&= \langle A\rangle_0 + \langle A\rangle_{1} - \langle A\rangle_{2} + \dots + (-1)^{\frac{(q+r-1)(q+r)}{2}}\langle A\rangle_{q+r}.
		\end{split}
	\end{equation}
	The last two operators of $\mathcal{G}_{q,r}$ that we will define in this section are the \textit{main grade involution operator} and the \textit{geometric conjugation operator}. Both are linear operators. In particular, the main grade involution operator is a linear operator that acts on $k$-vectors as:
	\begin{equation}
		\widehat{A}_k = (-1)^{k}A_k,
	\end{equation}
	and applied to an arbitrary multivector results in:
	\begin{equation}\label{maininvmultivector}
		\begin{split}
			\widehat{A} &= \langle\widehat{A}\rangle_{0} + \dots + \langle\widehat{A}\rangle_{q+r}\\
			&= \langle A\rangle^{\widehat{\empty}}_{0} + \dots + \langle A\rangle^{\widehat{\empty}}_{q+r}\\
			&= \langle A\rangle_0 - \langle A\rangle_{1} + \langle A\rangle_{2} + \dots + (-1)^{q+r}\langle A\rangle_{q+r}.
		\end{split}
	\end{equation}
	Finally, the \textit{geometric conjugation operator} acts on $k$-vectors as:
	\begin{equation}
		\overline{A}_k = (-1)^{k(k+1)/2}A_k,
	\end{equation}
	and on a arbitrary multivector as:
	\begin{equation}\label{conjugatemultivector}
		\begin{split}
			\overline{A} &= \langle\overline{A}\rangle_{0} + \dots + \langle\overline{A}\rangle_{q+r}\\
			&= \langle A\rangle^{-}_{0} + \dots + \langle A\rangle^{-}_{q+r}\\
			&= \langle A\rangle_0 - \langle A\rangle_{1} - \langle A\rangle_{2} + \dots + (-1)^{\frac{(r+q)(r+q+1)}{2}}\langle A\rangle_{q+r}.
		\end{split}
	\end{equation}

	According to \cite{Hes84}, the reversion operator can be used to defined the \textit{norm or magnitude} of an arbitrary multivector $A$:
	\begin{equation}\label{norm}
		|A| = \sqrt{\langle A\widetilde{A}\rangle_0},
	\end{equation}
	where $\langle\cdot\rangle_0$ is the grade-0 projection operator. However, \cite{Hitzer_inv} found a set of formulas involving various operators defined in this section that, while they do not define norms in $\mathcal{G}_{q,r}$ for $q+r\leq5$, they do constitute a scalar measure of multivectors, as they are used to normalize their inverse, when it exists. In particular, given a multivector $A\in\mathcal{G}_{q,r}$ such that $q+r=2$, the scalar measure of $A$ is $|A|=A\overline{A}$, i.e., the product of $A$ by its geometric conjugate, while given a multivector $B\in\mathcal{G}_{q,r}$ such that $q+r=3$, the scalar measure of $B$ is $|B|=B\overline{B}\widehat{B}\widetilde{B}$, i.e., the product of $B$ by its geometric conjugate, main involution, and reverse involution. We will adopt these scalar measures throughout the paper, especially to prove the existence theorem \ref{existence_theorem}.

	Finally, as mentioned in Section \ref{Sec:Intro}, both the definition and properties of the geometric Laplace transform are given in $\mathcal{G}_2$ to simplify the notation and computations. However, as will become evident throughout the manuscript, analogous reasoning can be applied if the geometric Laplace transform is considered in $\mathcal{G}_{q,r}$ for $q+r\leq5$.

	\section{The geometric Laplace transform}\label{Sec:GLT}
	
	Before giving the definition of the geometric Laplace transform, let us state the type of functions we are interested in transforming with this new Laplace transform. 
	
	\begin{defi}
		Let $ f:\mathbb{R}\to\mathcal{G}_2 $ be a multivector-valued function of a real variable. If there exist functions $ f_i:\mathbb{R}\to\mathbb{R} $ for $ i = 0,1,2,3 $ such that $ f(t) = f_0(t) + f_1(t)e_1 + f_2(t)e_2 + f_3(t)e_{12} $, then $ f $ is said to be a \textit{component-wise function}.
	\end{defi}
	The general case is completely analogous. Since there are $2^{q+r}$ basis elements in $\mathcal{G}_{q,r}$, a component-wise function $ f:\mathbb{R}\to\mathcal{G}_{q,r} $ satisfies that there exist real-valued functions of a real variable $ f_1, \dots, f_{2^{q+r}} $ such that $ f(t) = f_1(t) + f_2(t)e_1 + \dots + f_{2^{q+r}}(t)e_{1}\wedge\dots\wedge e_{q+r} $.
	
	\begin{rem}
		All continuous multivector-valued functions of a real variable admit a component-wise representation. In addition, if they are differentiable (in the classical sense), their Taylor expansions can be used to compute their component-wise representations:
		\begin{itemize}
			\item The function $ f(t)=\exp(ae_{12}t) $ with $ a\in\mathbb{R} $ can be easily represented as the component-wise function $ f(t)=\cos(at)+\sin(at)e_{12} $, where $ f_0(t)=\cos(at) $, $ f_1\equiv f_2\equiv 0 $, and $ f_3(t) = \sin(at) $.
			\item The function $ f(t)=\sin(ae_1t) $ with $ a\in\mathbb{R} $ can be represented as the component-wise function $ f(t)=\sin(at)e_1 $, where $ f_0\equiv f_2\equiv f_3\equiv 0 $, and $ f_1(t)=\sin(at) $.
			\item The function $ f(t) =\exp(ae_2t) $ with $ a\in\mathbb{R} $ can be represented as the component-wise function:
			\begin{equation}
				\begin{split}
					f(t)&=\left(\frac{{e}^{at}}{2}+\frac{{e}^{-at}}{2}\right) + \left(\frac{{e}^{at}}{2}-\frac{{e}^{-at}}{2}\right)e_{2}\\
					&= \cosh(at)+\sinh(at)e_2,
				\end{split}
			\end{equation} 
			where $ f_1\equiv f_3\equiv 0 $, $ f_0(t) = \cosh(at)$, and $f_2(t)=\sinh(at)$.
		\end{itemize}
	\end{rem}
	Now, the geometric Laplace transform is defined in two different ways.
	
	\begin{defi}\label{GLP_def}
		Let $f:\mathbb{R}\to\mathcal{G}_2$ be a component-wise multivector-valued function of a real variable. If exists, its \textit{left geometric Laplace transform}, denoted by $\mathcal{L}_l\left\{f\right\}$, is defined as the multivector-valued function of a multivector variable:
		\begin{equation}
			\mathcal{L}_l\left\{f\right\}(p)=\int_0^{\infty}e^{-pt}f(t)dt
		\end{equation}
		Analogously, if exists, the \textit{right geometric Laplace transform} of $f$, denoted by $\mathcal{L}_r\left\{f\right\}$, is defined as the multivector-valued function of a multivector variable:
		\begin{equation}
			\mathcal{L}_r\left\{f\right\}(p)=\int_0^{\infty}f(t)e^{-pt}dt
		\end{equation}
	\end{defi}
	
	\begin{rem}
		\begin{itemize}
			\item[\empty] \empty
			\item Notice that, compared with the classical definition of the Laplace transform and with earlier works by \cite{Cai2017} and \cite{Bau2023}, $ p \in \mathcal{G}_2 $ (or, in general, $p\in\mathcal{G}_{q,r}$), i.e., $p$ is a multivector variable and no longer a quaternionic or complex variable. Therefore, the geometric Laplace transform defined in \ref{GLP_def} generalizes the classical and quaternionic Laplace transforms reviewed in Section \ref{Sec:Intro}. Indeed, when $ \mathcal{G}_{q,r}$ is substituted by the complex numbers or the quaternions (which are sub-algebras of $ \mathcal{G}_2 $ and $ \mathcal{G}_3 $, respectively), the proposed definition matches the already existing ones.
			\item Since $ p $ is a multivector variable, the expression $ e^{-pt} $ is always a multivector and, therefore, $ f(t) $ does not, in general, commute with it. This is the reason why the geometric Laplace transform is defined in two different ways (left and right).
		\end{itemize}
	\end{rem}

	\section{Existence of the geometric Laplace transform and fundamental properties}\label{Sec:FP}
	
	Before enumerating the different properties of the geometric Laplace transform, we need to study the conditions under which the left and right geometric Laplace transforms of a component-wise multivector-valued function of a real variable $ f $ exist.
	
	\begin{defi}
		A component-wise multivector-valued function of a real variable $ f $ is said to be \textit{convergent} if:
		\begin{equation}
			\int_{a}^{b} f(t) \, dt
		\end{equation}
		converges for $ -\infty \leq a \leq b \leq +\infty $. It is said to be \textit{absolutely convergent} if:
		\begin{equation}
			\int_{a}^{b} |f(t)| \, dt
		\end{equation}
		converges for $ -\infty \leq a \leq b \leq +\infty $. Here, the norm $ |\cdot| $ is either the multivector norm defined in equation \eqref{norm} or the adopted scalar measure introduced in \citep{Hitzer_inv}.
	\end{defi}
	Clearly, an absolutely convergent function is also convergent.
	
	\begin{defi}
		A component-wise multivector-valued function of a real variable $f$ is said to have exponential order $\alpha\in\mathbb{R}$ if there exists a constant $M>0$ such that:
		\begin{equation}
			|f(t)| \leq Me^{\alpha t}
		\end{equation}
		for all $t\geq 0$.
	\end{defi}

	\begin{lem}\label{lemma1}
		For a multivector $p = p_0 + p_1e_1+p_2e_2+p_3e_{12}\in\mathcal{G}_2$, the following holds:
		\begin{equation}
			|e^{-p}|=e^{-2p_0},
		\end{equation}
		where $|\cdot|$ is the scalar measure introduced in \citep{Hitzer_inv}, i.e., for a multivector $A\in\mathcal{G}_2$, $|A|=A\overline{A}$ with $\overline{A}$ the geometric conjugate of $A$.
	\end{lem}
	\begin{proof}
		It is sufficient to prove that, for a multivector $p\in\mathcal{G}_2$, $\overline{e^{-p}}=e^{-\overline{p}}$. Indeed, if that is the case, then:
		\begin{equation}
			|e^{-p}|=e^{-p}\overline{e^{-p}} = e^{-p}e^{-\overline{p}}.
		\end{equation}
		But now, since $p$ and $\overline{p}$ always commute (recall that, as stated in Section \ref{Sec:IntroGA}, $p\overline{p}=\overline{p}p\in\mathbb{R}$), it follows that $e^{-p}e^{-\overline{p}}=e^{-(p+\overline{p})}$ which is equal to $e^{-2p_0}$ if $p = p_0 + p_1e_1+p_2e_2+p_3e_{12}$.
		
		To prove that $\overline{e^{-p}}=e^{-\overline{p}}$, we proceed as follows. Since $p = p_0 + p_1e_1+p_2e_2+p_3e_{12}$, we have that:
		\begin{equation}
			e^{-p} = e^{-(p_0 + p_1e_1+p_2e_2+p_3e_{12})} = e^{-p_0}e^{-(p_1e_1+p_2e_2+p_3e_{12})}.
		\end{equation}
		Now, since $(-(p_1e_1+p_2e_2+p_3e_{12}))^2 = p_1^2 + p_2^2 - p_3^2\in\mathbb{R}$:
		\begin{equation}\label{exp_series}
			\begin{split}
				e^{\bm{p}_v} &= \sum\limits_{n=0}^\infty \dfrac{\bm{p}_v^n}{n!}=\\
				&= 1 - \dfrac{\bm{p}_v}{1!} + \dfrac{p_v^2}{2!} - \dfrac{p_v^2\bm{p}_v}{3!} +  \dots
			\end{split}
		\end{equation}
		where $\bm{p}_v=-(p_1e_1+p_2e_2+p_3e_{12})$ and $p_v = \sqrt{p_1^2 + p_2^2 - p_3^2}$. Depending on the sign of $p_v^2$, the series defined in equation \eqref{exp_series} converges to:
		\begin{equation}\label{exp_expanded}
			\left\{\begin{split}
				\cos(p_v') - \frac{\bm{p}_v}{p_v'}\sin(p_v')&\qquad\text{if }p_v^2<0\\
				\empty\\
				\cosh(p_v) - \frac{\bm{p}_v}{p_v}\sinh(p_v)&\qquad\text{if }p_v^2>0,
			\end{split}\right.
		\end{equation}
		where $p_v'= \sqrt{p_3^2-p_1^2-p_2^2}$.	In both cases, since the geometric conjugation operator is linear, it can be observed that:
		\begin{equation}
			\begin{split}
				\overline{e^{-p}} &= \overline{e^{-p_0}e^{\bm{p}_v}}\\
				&=\overline{e^{-p_0}\left(\cosh(p_v) - \frac{\bm{p}_v}{p_v}\sinh(p_v)\right)}\\
				&=e^{-p_0}\left(\cosh(p_v) - \frac{\overline{\bm{p}_v}}{p_v}\sinh(p_v)\right)\\
				&=e^{-p_0}e^{\overline{\bm{p}}_v} = e^{-\overline{p}},
			\end{split}
		\end{equation}
		which completes the proof (the case with trigonometric functions is completely analogous).
	\end{proof}
	\begin{rem}
		The result remains valid if, for any scalar $\delta \in \mathbb{R}$, we consider $e^{-p\delta}$ instead. In this case, $|e^{-p\delta}| = e^{-2p_0\delta}$.
	\end{rem}
	\begin{rem}\label{remark1}
		The same result holds for the algebras $\mathcal{G}_{1,1}$ and $\mathcal{G}_{0,2}$, with a proof that closely follows the reasoning in Lemma \ref{lemma1}. 
		
		Moreover, this lemma can be extended to any geometric algebra $\mathcal{G}_{q,r}$ such that $q+r = 3$. For these algebras, $|e^{-p}| = e^{-4p_0}$. Although the proof is similar to that of Lemma \ref{lemma1}, it is more technically involved. The key steps include first showing that for any multivector $p$, the following identities hold: $\overline{e^{-p}} = e^{-\overline{p}}$, $\widehat{e^{-p}} = e^{-\widehat{p}}$, and $\widetilde{e^{-p}} = e^{-\widetilde{p}}$. From this, it remains to demonstrate that if $p = p_0 + p_1 e_1 + p_2 e_2 + p_3 e_3 + p_4 e_{12} + p_5 e_{13} + p_6 e_{23} + p_7 e_{123} \in \mathcal{G}_{q,r}$, then $e^{-(p + \overline{p} + \widehat{p} + \widetilde{p})} = e^{-4p_0}$. But this follows directly since $p + \overline{p} = 2p_0 + 2p_7 e_{123}$ and $\widehat{p} + \widetilde{p} = 2p_0 - 2p_7 e_{123}$.
		
		The other cases, i.e., geometric algebras $\mathcal{G}_{q,r}$ such that $q+r=4$ or $q+r=5$, are completely analogous but even more technically involved. Based on the formulas developed by \cite{Hitzer_inv}, generalizations of lemma \ref{lemma1} for those algebras can also be stated and proved.
	\end{rem}
	Now, we are in a position to state the conditions under which the geometric Laplace transform of a component-wise multivector-valued function of a real variable $f$ exists.
	\begin{theo}\label{existence_theorem}
		Let $f:\mathbb{R}\to\mathcal{G}_{2}$ be a component-wise multivector-valued function of a real variable defined for all $t\geq 0$. In addition, $f$ has exponential order $\alpha\in\mathbb{R}$. Then, its right and left geometric Laplace transforms exist as multivector-valued functions of a multivector variable for all multivectors $p\in\mathcal{G}_2$ satisfying that $\langle p\rangle_0>\alpha/2$.
	\end{theo}
	\begin{proof}
		As for the classical Laplace transform, we will use the fact that $f$ has exponential order $\alpha$ to prove that both $f(t)e^{-pt}$ and $e^{-pt}f(t)$ are absolutely convergent for $a=0$ and $b=+\infty$. Then, the conclusion will follow immediately. 
		
		Since both cases are completely analogous, we have that:
		\begin{align*}
			&\hspace{0.45cm}\int_0^{+\infty}|f(t)e^{-pt}|dt && \\
			&=\int_0^{+\infty}|f(t)||e^{-pt}|dt && |e^{-pt}|=e^{-2p_0t}\text{ (lemma \ref{lemma1})}\\
			&=\int_0^{+\infty}|f(t)|e^{-2p_0t} dt && f\text{ has exponential order} \\
			&\leq\int_0^{+\infty}Me^{\alpha t}e^{-2p_0t} dt && \\
			&=\int_0^{+\infty}Me^{-(2p_0-\alpha)t}dt &&
		\end{align*}
		Now:
		\begin{equation}
			\begin{split}
				\int_0^{+\infty}Me^{-(2p_0-\alpha)t}dt &= \lim\limits_{n\to\infty}\int_0^nMe^{-(2p_0-\alpha)t}dt\\
				&=\lim\limits_{n\to\infty}\left[\dfrac{Me^{-(2p_0-\alpha)t}}{-(2p_0-\alpha)}\right]_0^n\\
				&=\lim\limits_{n\to\infty}\left(\dfrac{Me^{-(2p_0-\alpha)n}}{-(2p_0-\alpha)} + \dfrac{M}{2p_0-\alpha}\right)\\
				&= \dfrac{M}{2p_0-\alpha},
			\end{split}
		\end{equation}
		where the last step is true only if $\langle p\rangle_0>\alpha/2$. In turn, this proves that $f(t)e^{-pt}$ is absolutely convergent, which as stated before, implies the conclusion.
	\end{proof}

	\begin{rem}		
		The condition $\langle p\rangle_0 > \alpha/2$ is equivalent to the requirement stated in the existence theorem for the classical Laplace transform. In particular, if the function $f$ has exponential order $\alpha$, the classical Laplace transform exists as a complex-valued function of a complex variable for all complex numbers $s \in \mathbb{C}$ such that $\text{Re}(s) > \alpha$.
		
		Interestingly, as noted in remark \ref{remark1}, the existence theorem for the geometric Laplace transform in geometric algebras $\mathcal{G}_{q,r}$ such that $q+r=3$ states that such a transform exists as a multivector-valued function of a multivector variable for all multivectors satisfying $p_0 > \alpha/4$. Similar conditions also arise for geometric algebras $\mathcal{G}_{q,r}$ where $q+r=4$ or $q+r=5$. The proof of these generalizations is not included here since it is completely analogous to that of Theorem \ref{existence_theorem}.
	\end{rem}

	\begin{rem}
		From now on, and unless explicitly stated, all functions considered will be component-wise multivector-valued functions of a real variable with exponential order. Therefore, we will omit explicitly writing this every time a function of this kind is used in the manuscript.
	\end{rem}
	
	\begin{cor}\label{real_function}
		If $f:\mathbb{R}\to\mathcal{G}_2$ is a real-valued function of a real variable, i.e., if the $f_1,f_2$ and $f_3$ components are zero, then its left and right geometric Laplace transforms are:
		\begin{equation}
			\begin{split}
				\mathcal{L}_r\left\{f\right\}(p)&=\int_0^{\infty}f(t)e^{-pt}dt = \mathcal{L}\{f\}(p)\\
				\mathcal{L}_l\left\{f\right\}(p)&=\int_0^{\infty}e^{-pt}f(t)dt = \int_0^{\infty}f(t)e^{-pt}dt=\mathcal{L}\{f\}(p).
			\end{split}
		\end{equation}
		In other words, for real-valued functions of a real variable, the left and right geometric Laplace transforms coincide and are the same as the classical Laplace transform, with a slight difference: the classical Laplace transform is a function of a complex variable, while the geometric Laplace transform, even in this case, is a function of a multivector variable.
	\end{cor}
	
	\subsection{Fundamental properties of the geometric Laplace transform}
	
	\subsubsection{Linearity of the geometric Laplace transform}\label{linearity}
	
	Due to the lack of commutativity, we can only discuss either right-linearity or left-linearity (depending on which side the constants are multiplied from). Both results are entirely analogous, so only right-linearity is proven.
	
	\begin{theo}
		Let $f,g:\mathbb{R}\to\mathcal{G}_2$ be two functions such that $\mathcal{L}_l\left\{f\right\}(p)$ and $\mathcal{L}_l\left\{g\right\}(p)$ exist. Let $A$ and $B$ be two multivectors of $\mathcal{G}_2$. Then:
		\begin{equation}
			\mathcal{L}_l\left\{fA +gB\right\}(p)=\mathcal{L}_l\left\{f\right\}(p)A+\mathcal{L}_l\left\{g\right\}(p)B
		\end{equation}
	\end{theo}	
		
	\begin{proof}
		The proof is a straightforward computation. Indeed:
		\begin{equation}
			\begin{split}
			\mathcal{L}_l\left\{fA +gB\right\}(p)&=\int_0^{\infty}e^{-pt}\left(f(t)A +g(t)B\right)dt\\ 
			&=\int_0^{\infty}e^{-pt}f(t)Adt + e^{-pt}g(t)Bdt\\
			&=\left(\int_0^{\infty}e^{-pt}f(t)dt\right)A + \left(\int_0^{\infty}e^{-pt}g(t)dt\right)B\\
			&=\mathcal{L}_l\left\{f\right\}(p)A +\mathcal{L}_l\left\{g\right\}(p)B
			\end{split}
		\end{equation}
		where the last equality is guaranteed as both Laplace transforms exist.
	\end{proof}
	
	In addition, using right or left-linearity, we can explicitly express the geometric Laplace transform of a component-wise function $f$ in terms of the classical Laplace transform as follows:
	\begin{equation}
		\begin{split}
		\mathcal{L}_l\left\{f\right\}(p)&=\int_0^{\infty}e^{-pt}f(t)dt\\
		&=\int_0^{\infty}e^{-pt}\left(f_0(t)+f_1(t)e_1+f_2(t)e_2+f_{3}(t)e_{12}\right)dt\\
		&=\int_0^{\infty}e^{-pt}f_0(t)+e^{-pt}f_1(t)e_1+e^{-pt}f_2(t)e_2+e^{-pt}f_{3}(t)e_{12}dt\\
		&=\int_0^{\infty}e^{-pt}f_0(t)dt+\int_0^{\infty}e^{-pt}f_1(t)e_1dt+\int_0^{\infty}e^{-pt}f_2(t)e_2dt+\int_0^{\infty}e^{-pt}f_{3}(t)e_{12}dt.
		\end{split}\nonumber
	\end{equation}
	Now, the basis elements $\{1,e_1,e_2,e_{12}\}$ are constants with respect to variable $t$ and, therefore, can be factored out:
	\begin{equation}
		\begin{split}
		\mathcal{L}_l\left\{f\right\}(p) &=\int_0^{\infty}e^{-pt}f(t)dt\\
		&=\int_0^{\infty}e^{-pt}f_0(t)dt+\int_0^{\infty}e^{-pt}f_1(t)dt e_1+\int_0^{\infty}e^{-pt}f_2(t)dt e_2+\int_0^{\infty}e^{-pt}f_{3}(t)dt e_{12}.
		\end{split}\nonumber
	\end{equation}
	These four integrals can be computed using corollary \ref{real_function}, which allows us to obtain:
	\begin{equation}\label{linearity_eq}
		\mathcal{L}_l\left\{f\right\}(p)=F_0(p)e_0+F_1(p)e_1+F_2(p)e_2+F_{3}(p)e_{12}
	\end{equation}
	where $F_0, F_1, F_2$, and $F_3$ are multivector-valued functions of a multivector variable. 
	
	Therefore, for component-wise functions, the problem of computing the geometric Laplace transform reduces to that of computing the geometric Laplace transform of four ($2^{q+r}$ in the general case) real-valued functions of a real variable.
	\begin{rem}
		It is easy to see that, for the same component-wise function $t$, $\mathcal{L}_r\left\{f\right\}(p) = e_0F_0(p)+e_1F_1(p)+e_2F_2(p)+e_{12}F_3(p)$ (again by corollary \ref{real_function}).
	\end{rem}
	
	\subsubsection{Geometric Laplace transform of the derivative}\label{derivative1}
	
	Let us denote by $F_l(p)$ the left geometric Laplace transform of the component-wise function $f$, i.e., $\mathcal{L}_l\left\{f\right\}(p)=F_l(p)$, if it exists.
	
	\begin{theo}\label{theo_derivative1}
		Let $f:\mathbb{R}\to\mathcal{G}_2$ be a derivable component-wise function such that $\mathcal{L}_l\left\{f\right\}(p)$ exists. Then:
		\begin{equation}
			\mathcal{L}_l\left\{\frac{d}{dt}f\right\}(p)=pF_l(p)-f(0)
		\end{equation}
	\end{theo}
	
	Before proving theorem \ref{theo_derivative1}, the following technical lemma is needed.
	\begin{lem}\label{lemma_integration_by_parts}
		Let $f,g:\mathbb{R}\to\mathcal{G}_2$ be two derivable component-wise multivector-valued functions of a real variable. Then, the product rule holds for $f$ and $g$, i.e., $\dfrac{d}{dt}\left(fg\right) = \left(\dfrac{d}{dt}f\right)g + f\left(\dfrac{d}{dt}g\right)$.
	\end{lem}
	
	\begin{proof}\mbox{}\\
		
		
		
		If $f(t) = f_0(t)+f_1(t)e_1+f_2(t)e_2+f_3(t)e_{12}$ and $g(t) = g_0(t)+g_1(t)e_1+g_2(t)e_2+g_3(t)e_{12}$ with $f_i,g_i:\mathbb{R}\to\mathbb{R}$ for $i=0,1,2,3$, then:
		\begin{equation}
			\begin{split}
				\dfrac{d}{dt}f(t) &=f'(t) = f'_0+f'_1e_1+f'_2e_2+f'_3e_{12}\\
				\dfrac{d}{dt}g(t) &=g'(t) = g'_0+g'_1e_1+g'_2e_2+g'_3e_{12},
			\end{split}
		\end{equation}
		where the notation for the derivatives is simplified, as all derivatives are with respect to the real variable $t$. Now:
		\begin{equation}
			\begin{split}
				fg &= (f_0+f_1e_1+f_2e_2+f_3e_{12})(g_0+g_1e_1+g_2e_2+g_3e_{12})\\
				&=\left(f_{0}g_{0}+f_{1}g_{1}+f_{2}g_{2}-f_{3}g_{3}\right)\\ 
				&+\left(f_{0}g_{1}+f_{1}g_{0}-f_{2}g_{3}+f_{3}g_{2}\right)e_{1}\\  
				&+\left(f_{0}g_{2}+f_{2}g_{0}+f_{1}g_{3}-f_{3}g_{1}\right)e_{2}\\
				&+\left(f_{0}g_{3}+f_{1}g_{2}-f_{2}g_{1}+f_{3}g_{0}\right)e_{12}.
			\end{split}\nonumber
		\end{equation}
		
		Taking derivatives and applying the product rule to the real-valued functions of a real variable $f_i$ and $g_i$ (for $i=0,1,2,3$), we have:
		\begin{equation}
			\begin{split}
				(fg)' &=\left(f'_{0}g_{0}+f_{0}g'_{0}+f'_{1}g_{1}+f_{1}g'_{1}+f'_{2}g_{2}+f_{2}g'_{2}-f'_{3}g_{3}-f_{3}g'_{3}\right)\\ 
				&+\left(f'_{0}g_{1}+f_{0}g'_{1}+f'_{1}g_{0}+f_{1}g'_{0}-f'_{2}g_{3}-f_{2}g'_{3}+f'_{3}g_{2}+f_{3}g'_{2}\right)e_{1}\\  
				&+\left(f'_{0}g_{2}+f_{0}g'_{2}+f'_{2}g_{0}+f_{2}g'_{0}+f'_{1}g_{3}+f_{1}g'_{3}-f'_{3}g_{1}-f_{3}g'_{1}\right)e_{2}\\
				&+\left(f'_{0}g_{3}+f_{0}g'_{3}+f'_{1}g_{2}+f_{1}g'_{2}-f'_{2}g_{1}-f_{2}g'_{1}+f'_{3}g_{0}+f_{3}g'_{0}\right)e_{12}.
			\end{split}\nonumber
		\end{equation}
		On the other hand, we have:
		\begin{equation}
			\begin{split}
				f'g &= \left(f_{0}'g_{0}+f_{1}'g_{1}+f_{2}'g_{2}-f_{3}'g_{3}\right)\\ 
				&+\left(f_{0}'g_{1}+f_{1}'g_{0}-f_{2}'g_{3}+f_{3}'g_{2}\right)e_{1}\\  
				&+\left(f_{0}'g_{2}+f_{2}'g_{0}+f_{1}'g_{3}-f_{3}'g_{1}\right)e_{2}\\
				&+\left(f_{0}'g_{3}+f_{1}'g_{2}-f_{2}'g_{1}+f_{3}'g_{0}\right)e_{12},
			\end{split}\text{ and }\begin{split}
			fg' &= \left(f_{0}g'_{0}+f_{1}g'_{1}+f_{2}g'_{2}-f_{3}g'_{3}\right)\\ 
			&+\left(f_{0}g'_{1}+f_{1}g'_{0}-f_{2}g'_{3}+f_{3}g'_{2}\right)e_{1}\\  
			&+\left(f_{0}g'_{2}+f_{2}g'_{0}+f_{1}g'_{3}-f_{3}g'_{1}\right)e_{2}\\
			&+\left(f_{0}g'_{3}+f_{1}g'_{2}-f_{2}g'_{1}+f_{3}g'_{0}\right)e_{12}
			\end{split}\nonumber
		\end{equation}
		So that:
		\begin{equation}
			\begin{split}
				f'g+fg' &= \left(f_{0}'g_{0}+f_{1}'g_{1}+f_{2}'g_{2}-f_{3}'g_{3}+f_{0}g'_{0}+f_{1}g'_{1}+f_{2}g'_{2}-f_{3}g'_{3}\right)\\ 
				&+\left(f_{0}'g_{1}+f_{1}'g_{0}-f_{2}'g_{3}+f_{3}'g_{2}+f_{0}g'_{1}+f_{1}g'_{0}-f_{2}g'_{3}+f_{3}g'_{2}\right)e_{1}\\  
				&+\left(f_{0}'g_{2}+f_{2}'g_{0}+f_{1}'g_{3}-f_{3}'g_{1}+f_{0}g'_{2}+f_{2}g'_{0}+f_{1}g'_{3}-f_{3}g'_{1}\right)e_{2}\\
				&+\left(f_{0}g'_{3}+f_{1}g'_{2}-f_{2}g'_{1}+f_{3}g'_{0}+f_{0}g'_{3}+f_{1}g'_{2}-f_{2}g'_{1}+f_{3}g'_{0}\right)e_{12},
			\end{split}\nonumber
		\end{equation}
		which clearly coincides with $(fg)'$.
		
		
	\end{proof}
	
	Lemma \ref{lemma_integration_by_parts} allows us to define the integration by parts method on component-wise multivector-valued functions of a real variable, so we are in conditions to prove theorem \ref{theo_derivative1}.
	
	\begin{proof}
		The proof is straightforward integrating by parts:
		\begin{align*}
			\mathcal{L}_l\left\{\frac{d}{dt}f\right\}(p)&=\int_0^{\infty}e^{-pt}\frac{d}{dt}f(t)dt\\
			&=e^{-pt}f(t)\bigg|_0^\infty+\int_0^{\infty}pe^{-pt}f(t)dt&& \left.\begin{cases}
				u(t) &= e^{-pt}\\
				\frac{d}{dt}v(t) &= \frac{d}{dt}f(t)
			\end{cases}\right\}\\
			&=-f(0)+pF_l(p).\\
		\end{align*}
	\end{proof}
	
	\subsubsection{Geometric Laplace transform of the second and $n$-th derivative}
	
	We first prove the general case, i.e., the expression for the geometric Laplace transform of the $n$-th derivative, and then deduce the expression for the geometric Laplace transform of the second derivative.
	
	\begin{theo}
		For every $n > 0$, let $f: \mathbb{R} \to \mathcal{G}_2$ be a component-wise $n$-times differentiable function such that $\mathcal{L}_l\left\{f\right\}(p)$ exists. Then:
		\begin{equation}
			\mathcal{L}_l\left\{\frac{d^n}{dt^n}f\right\}(p)=p^nF_l(p)-\sum_{k=1}^{n}p^{n-k}\frac{d^{k-1}}{dt^{k-1}}f(t)\bigg|_{t=0}
		\end{equation}
	\end{theo}
	
	\begin{proof}
		The proof proceeds by induction on $n$.
		\begin{itemize}
			\item For $n=1$, the formula coincides with that of Theorem \ref{theo_derivative1} and is, therefore, true.
			\item Assuming the formula holds for $1,2,\dots,n$, it can be proved that it also holds for $n+1$. Indeed:
			{\small \begin{align*}
				\mathcal{L}_l\left\{\frac{d^{n}+1}{dt^{n+1}}f\right\}(p)&=\int_0^{\infty}e^{-pt}\frac{d^{n+1}}{dt^{n+1}}f(t)dt\\
				&=e^{-pt}\frac{d^{n}}{dt^{n}}f(t)\bigg|_0^\infty+\int_0^{\infty}pe^{-pt}\frac{d^{n}}{dt^{n}}f(t)dt && \left.\begin{cases}
					u(t) &= e^{-pt}\\
					\frac{d}{dt}v(t) &= \frac{d^{n+1}}{dt^{n+1}}f(t)
				\end{cases}\right\}\\
				&=e^{-pt}\frac{d^{n}}{dt^{n}}f(t)\bigg|_0^\infty+p\int_0^{\infty}e^{-pt}\frac{d^{n}}{dt^{n}}f(t)dt\\
				&=-\frac{d^{n}}{dt^{n}}f(t)\bigg|_{t=0}+p^{n+1}F_l(p)-\sum_{k=1}^{n}p^{n+1-k}\frac{d^{k-1}}{dt^{k-1}}f(t)\bigg|_{t=0} && \text{Induction hypothesis}\\
				&=p^{n+1}F_l(p)-\sum_{k=1}^{n+1}p^{n+1-k}\frac{d^{k-1}}{dt^{k-1}}f(t)\bigg|_{t=0},
			\end{align*}}
			which completes the proof.
		\end{itemize}
	\end{proof}
	
	Finally, the expression for the geometric Laplace transform of the second derivative of a two-times differentiable component-wise function $f$ can be easily deduced from the general formula:
	\begin{equation}
		\mathcal{L}_l\left\{\frac{d^2}{dt^2}f\right\}(p)=p^2F_l(p)-pf(0)-\frac{d}{dt}f(t)\bigg|_{t=0}
	\end{equation}
	
	\subsubsection{Geometric Laplace transform of the time integral}\label{integral_property}
	
	\begin{theo}\label{integral_property_theorem}
		Let $f: \mathbb{R} \to \mathcal{G}_2$ be an integrable component-wise function such that $F_l(p)=\mathcal{L}_l\left\{f\right\}(p)$ exists. Then:
		\begin{equation}
			p\mathcal{L}_l\left\{\int_0^tf(u)du\right\}(p)=F_l(p)
		\end{equation}
	\end{theo}
	
	\begin{proof}
		Let us denote by $g$ the time integral of function $f$, i.e.,
		\begin{equation}
			g(t)=\int_0^tf(u)du.
		\end{equation}
		Since $g$ is clearly derivable, we have: 
		\begin{equation}
			\frac{d}{dt}g(t)=f(t)\quad\text{with}\quad g(0)=0.
		\end{equation}
		Now, applying the formula for the geometric Laplace transform of the derivative (theorem \ref{theo_derivative1}), we get:
		\begin{equation}
			\mathcal{L}_l\left\{\frac{d}{dt}g\right\}(p)=p\mathcal{L}_l\left\{g\right\}(p)-g(0)=p\mathcal{L}_l\left\{g\right\}(p)
		\end{equation}
		On the other hand, by definition:
		\begin{equation}
			\mathcal{L}_l\left\{\frac{d}{dt}g\right\}(p)=\mathcal{L}_l\left\{f\right\}(p)=F_l(p),
		\end{equation}
		so, by equating these two expressions, we obtain the desired result:
		\begin{equation}
			p\mathcal{L}_l\left\{g(t)\right\}=F_l(p).
		\end{equation}
	\end{proof}
	
	\subsubsection{Geometric Laplace transform of the frequency shift}\label{freq_shift} 
	
	\begin{theo}\label{frecuency_shift}
		Let $f:\mathbb{R}\to\mathcal{G}_2$ be a component-wise function such that $F_l(p)=\mathcal{L}_l\left\{f\right\}(p)$ exists. Then, for $a\in \mathbb{R}$:
		\begin{equation}
			\mathcal{L}_l\left\{e^{at}f\right\}(p)=F_l(p-a)
		\end{equation}
	\end{theo}
	
	\begin{proof}
		The proof is a straightforward computation. Indeed:
		\begin{align*}
			\mathcal{L}_l\left\{e^{at}f\right\}(p)&=\int_0^{\infty}e^{-pt}e^{at}f(t)dt\\
			&=\int_0^{\infty}e^{(a-p)t}f(t)dt && a\text{ and }p\text{ commute}\\
			&=\int_0^{\infty}e^{-ut}f(t)dt && u=p-a\\
			&=\int_0^{\infty}e^{-ut}f(t)dt=F_l(u)=F_l(p-a)
		\end{align*}
	\end{proof}
	
	\begin{rem}
		The result does not hold if $a\in\mathcal{G}_2$ since $a$ and $p$ do not commute in general.
	\end{rem}
	
	\subsubsection{Geometric Laplace transform of the time shift}\label{time_shift}
	
	\begin{theo}
		Let $f:\mathbb{R}\to\mathcal{G}_2$ be a component-wise function such that $F_l(p)=\mathcal{L}_l\left\{f\right\}(p)$ exists, and let $u:\mathbb{R}\to\mathbb{R}$ be the Heaviside function, i.e., $u(t)=1$ for $t\geq 0$ and $0$ otherwise. Then, for $a\in \mathbb{R}$, $a>0$:
		\begin{equation}
			\mathcal{L}_l\left\{f(t-a)u(t-a)\right\}(p)=e^{-ap}F_l(p)
		\end{equation}
	\end{theo}
	
	\begin{proof}
		By definition, the geometric Laplace transform of a time-shifted function is given by:
		\begin{equation}
			\mathcal{L}_l\left\{f(t-a)u(t-a)\right\}(p)=\int_0^{\infty}e^{-pt}f(t-a)u(t-a)dt.
		\end{equation}
		Now, performing the change of variables $v=t-a$, we have $t=v+a$ and $dt=dv$, so the integral becomes:
		\begin{equation}
			\mathcal{L}_l\left\{f(t-a)u(t-a)\right\}(p)=\int_{-a}^{\infty}e^{-p(v+a)}f(v)u(v)dv.
		\end{equation}
		Since the Heaviside function satisfies that $u(v)=0$ for $v<0$, the lower limit of integration can be extended from $-a$ to $0$. Now:
		\begin{align*}
			\mathcal{L}_l\left\{f(t-a)u(t-a)\right\}(p) &= \int_{0}^{\infty}e^{-p(v+a)}f(v)u(v)dv\\
			&= \int_{0}^{\infty}e^{-p(a+v)}f(v)u(v)dv\\
			&= \int_{0}^{\infty}e^{-pa}e^{-pv}f(v)u(v)dv && pa\text{ and }pv\text{ commute}\\
			&= e^{-pa}\int_{0}^{\infty}e^{-pv}f(v)u(v)dv && e^{-pa}\text{ does not depend on }v\\
			&= e^{-pa}\int_{0}^{\infty}e^{-pv}f(v)dv && u(v)\equiv 1\text{ for }v\geq0\\
			&= e^{-pa}F_l(p)
		\end{align*}
		which concludes the proof.
	\end{proof}
	
	A straightforward corollary can be stated for casual functions.
	\begin{cor}
		Let $f:\mathbb{R}\to\mathcal{G}_2$ be a component-wise function such that $F_l(p)=\mathcal{L}_l\left\{f\right\}(p)$ exists, and let $g:\mathbb{R}\to\mathbb{R}$ be a causal function, i.e., $g(t)=0$ for $t < 0$. Then, for $a\in \mathbb{R}$, $a>0$:
		\begin{equation}
			\mathcal{L}_l\left\{f(t-a)g(t-a)\right\}(p)=e^{-ap}\mathcal{L}_l\left\{fg\right\}(p)
		\end{equation}
	\end{cor}
	
	\subsubsection{Geometric Laplace transform of the time scaling}
	
	\begin{theo}
		Let $f:\mathbb{R}\to\mathcal{G}_2$ be a component-wise function such that $F_l(p)=\mathcal{L}_l\left\{f\right\}(p)$ exists. Then, for $a\in \mathbb{R}$, $a>0$:
		\begin{equation}
			\mathcal{L}_l\left\{f(at)\right\}(p)=\frac{1}{a}F_l\left(\frac{p}{a}\right)
		\end{equation}
	\end{theo}
	
	\begin{proof}
		The proof is similar to that of the time shift property. Indeed:
		\begin{align*}
			\mathcal{L}_l\left\{f(at)\right\}(p)&=\int_0^{\infty}e^{-pt}f(at)dt\\
			&=\int_0^{\infty}e^{-p\frac{v}{a}}f(v)\frac{dv}{a} && \text{change of variables }\left\{\begin{aligned}
				v&=at\\
				dv&=adt
			\end{aligned}\right\}\\
			&=\frac{1}{a}\int_0^{\infty}e^{-v\frac{p}{a}}f(v)dv = \frac{1}{a}F_l\left(\frac{p}{a}\right),
		\end{align*}
		where in the last identity, the definition of the geometric Laplace transform has been used.
	\end{proof}
	
	\subsubsection{Geometric Laplace transform of the convolution product}\label{Sec:Conv}
	
	To extend the well-known convolution property of the classical Laplace transform, we must first define the convolution operation for component-wise functions.
	\begin{defi}
		Let $f,g:\mathbb{R}\to\mathcal{G}_2$ be two multivector-valued, component-wise functions of a real variable. Then, the convolution product of $f$ and $g$, that we will called the geometric convolution, is defined as:
		\begin{equation}
			h(t) = (f\star g)(t) = \int_{-\infty}^{t} f(\tau) g(t - \tau)d\tau,
		\end{equation}
		where the product of $f$ and $g$ inside the integral is the geometric product. 
		
		In addition, if $f$ and $g$ are casual functions, i.e., $f(t)=g(t)=0$ for $t<0$, then:
		\begin{equation}
			h(t) = (f\star g)(t) = \int_{0}^{t} f(\tau) g(t - \tau)d\tau,
		\end{equation}
	\end{defi}

	Clearly, since $f$ and $g$ do not commute in general (as they are multivector-valued functions), their convolution product does not commute either. Therefore, the convolution product of component-wise functions is not commutative, in contrast to the convolution product of real-valued functions of a real variable. However, it is clear that if $f(t)g(t) = g(t)f(t)$ for all $t\in\mathbb{R}$, then their convolution product also commutes.
	
	\begin{rem}
		For other non-commutative algebras, such as the quaternions, the non-commutativity of their product leads to different definitions of the convolution product: the left-side convolution product \citep{ell_left_side_conv,ell_right_left_two_side_conv,chang_left_two_side_conv}, the right-side convolution product \citep{ell_right_left_two_side_conv}, and the two-sided convolution product \citep{chang_left_two_side_conv,sangwine_two_side_conv}. Each of these definitions can be extended to any of the geometric algebras considered in this work, i.e., $\mathcal{G}_{p,q}$ with $p+q \leq 5$. However, in this study, we adopt the left-side convolution product, as it aligns with the convention followed by most authors working with other non-commutative algebras.
	\end{rem}
	
	We are now in a position to state the extension of the convolution property of the classical Laplace transform to the case of component-wise functions.
	\begin{theo}
		Let $f,g:\mathbb{R}\to\mathcal{G}_2$ be two casual, component-wise functions such that $F_l(p)=\mathcal{L}_l\left\{f\right\}(p)$ and $G_l(p)=\mathcal{L}_l\left\{g\right\}(p)$ exist. Then:
		\begin{equation}
			\mathcal{L}_l\left\{(f\star g)(t)\right\}(p) = F_l(p)G_l(p)
		\end{equation}
	\end{theo}
	
	\begin{proof}
		By definition of the geometric Laplace transform, we have:
		\begin{equation}
			\begin{split}
				\mathcal{L}_l\left\{(f\star g)(t)\right\}(p)&=\mathcal{L}_l\left\{\int_0^t f(\tau)g(t-\tau)d\tau\right\}(p)\\
				&=\int_0^\infty e^{-pt}\int_0^t f(\tau)g(t-\tau)d\tau dt\\
				&=\int_0^\infty \int_0^t e^{-pt}f(\tau)g(t-\tau)d\tau dt
			\end{split}
		\end{equation}
		Introducing the change of variables $u = t-\tau$ and $v=\tau$, and adjusting the integration limits accordingly (noting that the Jacobian determinant of this transformation is 1), we obtain:
		\begin{equation}\label{change_of_variables}
			\begin{split}
					\int_0^\infty \int_0^t e^{-pt}f(\tau)g(t-\tau)d\tau dt&=\int_{u=0}^\infty\int_{v=0}^{\infty} e^{-p(v+u)}f(v)g(u)dvdu\\
				&=\int_{u=0}^\infty\int_{v=0}^{\infty} e^{-pv}e^{-pu}f(v)g(u)dvdu,
			\end{split}
		\end{equation}
		where the last equality holds because $pu$ and $pv$ commute.
		
		Since $f$ and $g$ are component-wise functions, we can expand the product $f(v)g(u)$ in terms of their component elements, resulting in:
		\begin{equation}\label{expand}
			\begin{split}
				fg &= (f_0+f_1e_1+f_2e_2+f_3e_{12})(g_0+g_1e_1+g_2e_2+g_3e_{12})\\
				&=\left(f_{0}g_{0}+f_{1}g_{1}+f_{2}g_{2}-f_{3}g_{3}\right)\\ 
				&+\left(f_{0}g_{1}+f_{1}g_{0}-f_{2}g_{3}+f_{3}g_{2}\right)e_{1}\\  
				&+\left(f_{0}g_{2}+f_{2}g_{0}+f_{1}g_{3}-f_{3}g_{1}\right)e_{2}\\
				&+\left(f_{0}g_{3}+f_{1}g_{2}-f_{2}g_{1}+f_{3}g_{0}\right)e_{12},
			\end{split}\nonumber
		\end{equation}
		where the dependence on the variables $u,v$ has been omitted for clarity. 
		
		Now, substituting this expansion into equation \eqref{change_of_variables}, and applying the linearity property of the geometric Laplace transform (Section \ref{linearity}), we can rewrite the integral as:
		\begin{equation}
			\begin{split}
				&\int_{u=0}^\infty\int_{v=0}^{\infty} e^{-pv}e^{-pu}f(v)g(u)dvdu\\
				&=\int_{u=0}^\infty \int_{v=0}^{\infty} e^{-pv}e^{-pu}\left(f_{0}g_{0}+f_{1}g_{1}+f_{2}g_{2}-f_{3}g_{3}\right)dvdu\\
				&+\int_{u=0}^\infty \int_{v=0}^{\infty} e^{-pv}e^{-pu}\left(f_{0}g_{1}+f_{1}g_{0}-f_{2}g_{3}+f_{3}g_{2}\right)e_{1}dvdu\\   
				&+\int_{u=0}^\infty \int_{v=0}^{\infty} e^{-pv}e^{-pu}\left(f_{0}g_{2}+f_{2}g_{0}+f_{1}g_{3}-f_{3}g_{1}\right)e_{2}dvdu\\   
				&+\int_{u=0}^\infty \int_{v=0}^{\infty} e^{-pv}e^{-pu}\left(f_{0}g_{3}+f_{1}g_{2}-f_{2}g_{1}+f_{3}g_{0}\right)e_{12}dvdu,
			\end{split}
		\end{equation}
		which, by linearity of the integral, simplifies to:
		\begin{equation}\label{split_integral}
			\begin{split}
				&\int_{0}^\infty \int_{0}^{\infty}e^{-pv}e^{-pu}f_{0}g_{0}dvdu+\int_{0}^\infty \int_{0}^{\infty}e^{-pv}e^{-pu}f_{1}g_{1}dvdu\\
				&+\int_{0}^\infty \int_{0}^{\infty}e^{-pv}e^{-pu}f_{2}g_{2}dvdu-\int_{0}^\infty \int_{0}^{\infty}e^{-pv}e^{-pu}f_{3}g_{3}dvdu\\
				&+\int_{0}^\infty \int_{0}^{\infty}e^{-pv}e^{-pu}f_{0}g_{1}e_{1}dvdu+\int_{0}^\infty \int_{0}^{\infty}e^{-pv}e^{-pu}f_{1}g_{0}e_{1}dvdu\\
				&-\int_{0}^\infty \int_{0}^{\infty}e^{-pv}e^{-pu}f_{2}g_{3}e_{1}dvdu+\int_{0}^\infty \int_{0}^{\infty}e^{-pv}e^{-pu}f_{3}g_{2}e_{1}dvdu\\   
				&+\int_{0}^\infty \int_{0}^{\infty}e^{-pv}e^{-pu}f_{0}g_{2}e_2dvdu+\int_{0}^\infty \int_{0}^{\infty}e^{-pv}e^{-pu}f_{2}g_{0}e_2dvdu\\
				&+\int_{0}^\infty \int_{0}^{\infty}e^{-pv}e^{-pu}f_{1}g_{3}e_2dvdu-\int_{0}^\infty \int_{0}^{\infty}e^{-pv}e^{-pu}f_{3}g_{1}e_{2}dvdu\\   
				&+\int_{0}^\infty \int_{0}^{\infty}e^{-pv}e^{-pu}f_{0}g_{3}e_{12}dvdu+\int_{0}^\infty \int_{0}^{\infty}e^{-pv}e^{-pu}f_{1}g_{2}e_{12}dvdu\\
				&-\int_{0}^\infty \int_{0}^{\infty}e^{-pv}e^{-pu}f_{2}g_{1}e_{12}dvdu+\int_{0}^\infty \int_{0}^{\infty}e^{-pv}e^{-pu}f_{3}g_{0}e_{12}dvdu.
			\end{split}
		\end{equation}
		Now, each term in equation \eqref{split_integral} is of the form:
		\begin{equation}
			\int_{u=0}^\infty \int_{v=0}^{\infty}e^{-pv}e^{-pu}f_{i}(v)g_{j}(u)e_{k}dvdu,
		\end{equation}
		for $0\leq i\neq j\leq 3$ and $k = 1,2,12$ (when $i=j$, the expression is the same except that there is no $e_k$). Therefore, we can treat each of these terms independently. 
		
		Since $f_i$ and $g_j$ are all real-valued functions of a real variable (and hence commute with any multivector) for $0\leq i,j\leq 3$, we can apply Fubini's theorem to each term, resulting in:
		\small{\begin{equation}
			\int_{u=0}^\infty \int_{v=0}^{\infty}e^{-pv}e^{-pu}f_{i}(v)g_{j}(u)e_{k}dvdu = \int_{v=0}^\infty \int_{u=0}^{\infty}e^{-pv}e^{-pu}f_{i}(v)g_{j}(u)e_{k}dudv.
		\end{equation}}
		Now, we can proceed as follows:
		\small{\begin{equation}\label{conv_def_laplace}
				\begin{split}
					\int_{v=0}^\infty \int_{u=0}^{\infty}e^{-pv}e^{-pu}f_{i}(v)g_{j}(u)e_{k}dudv&=\int_{v=0}^\infty e^{-pv}f_{i}(v)\int_{u=0}^{\infty}e^{-pu}g_{j}(u)e_{k}dudv\\
					&=\left(\int_{v=0}^\infty e^{-pv}f_{i}(v)\int_{u=0}^{\infty}e^{-pu}g_{j}(u)dudv\right)e_{k}\\
					&=\left(\int_{v=0}^\infty e^{-pv}f_{i}(v)dv\right)\left(\int_{u=0}^{\infty}e^{-pu}g_{j}(u)du\right)e_k\\
					&\stackrel{(1)}{=}\mathcal{L}_l\{f_i\}(p)\mathcal{L}_l\{g_j\}(p)e_k\\
					&= F_i(p)G_j(p)e_k,
				\end{split}
		\end{equation}}
		where we use the definition of the left geometric Laplace transform in identity $(1)$.
	
		Substituting equation \eqref{conv_def_laplace} into each term of equation \eqref{split_integral}, we have:
		\begin{equation}
			\begin{split}
				&F_{0}(p)G_{0}(p)+F_{1}(p)G_{1}(p)+F_{2}(p)G_{2}(p)-F_{3}(p)G_{3}(p)\\
				&+ \Big(F_{0}(p)G_{1}(p)+F_{1}(p)G_{0}(p)-F_{2}(p)G_{3}(p)+F_{3}(p)G_{2}(p)\Big)e_1\\
				&+\Big(F_{0}(p)G_{2}(p)+F_{2}(p)G_{0}(p)+F_{1}(p)G_{3}(p)-F_{3}(p)G_{1}(p)\Big)e_2\\
				&+\Big(F_{0}(p)G_{3}(p)+F_{1}(p)G_{2}(p)-F_{2}(p)G_{1}(p)+F_{3}(p)G_{0}(p)\Big)e_{12}.
			\end{split}
		\end{equation}
		This expression corresponds to the geometric product of $F_l(p)$ and $G_l(p)$. Therefore:
		\begin{equation}
			\mathcal{L}_l\left\{(f\star g)(t)\right\}(p)=F_l(p)G_l(p),
		\end{equation}
		which completes the proof.
	\end{proof}
		
	\subsubsection{Geometric Laplace transform of the main, reverse and geometric conjugate involutions}
	
	In Section \ref{Sec:IntroGA}, three different involutions in GA were introduced: the main involution, the reverse involution, and the geometric conjugate involution. We will now prove that applying an involution to a component-wise function $f$ and then taking the geometric Laplace transform yields the same result as applying the involution to the geometric Laplace transform of $f$. Although the three proofs are analogous, we will present each of them for the sake of completeness and formality.
	
	We start by stating the result for the main involution.
	
	\begin{theo}
		Let $f:\mathbb{R}\to\mathcal{G}_2$ be a component-wise function such that $F_l(p)=\mathcal{L}_l\left\{f\right\}(p)$ exists. Then:
		\begin{equation}
			\mathcal{L}_l\left\{\widehat{f}\right\}(p)=\widehat{F}_l\left(p\right),
		\end{equation}
		where $\widehat{f}$ denotes the main involution of $f$, as defined for general multivectors in \eqref{maininvmultivector}.
	\end{theo}
	
	\begin{proof}
		Since $f$ is a component-wise function, there exist functions $f_0,f_1,f_2,f_3:\mathbb{R}\to\mathbb{R}$ such that:
		\begin{equation}
			f(t)=f_0(t)+f_1(t)e_1+f_2(t)e_2+f_3(t)e_{12}
		\end{equation}
		and, therefore:
		\begin{equation}
			\widehat{f}(t)=f_0(t)-f_1(t)e_1-f_2(t)e_2+f_3(t)e_{12}
		\end{equation}
		As stated in Section \ref{linearity}, equation \eqref{linearity_eq}, the right and left linearity of the geometric Laplace transform imply that the transform of a component-wise function $f$ can be expressed as:
		\begin{equation}
			\mathcal{L}_l\left\{f\right\}(p)=F_0(p)+F_1(p)e_1+F_2(p)e_2+F_{3}(p)e_{12}
		\end{equation}
		where $F_0, F_1, F_2$, and $F_3$ are the geometric Laplace transforms of $f_0,f_1,f_2$ and $f_3$, respectively. Now, applying also this to $\widehat{f}$, we obtain:
		\begin{equation}
			\mathcal{L}_l\left\{\widehat{f}\right\}(p)=F_{0}(p)-F_{1}(p)e_1-F_{2}(p)e_2+F_{3}(p)e_{12}.
		\end{equation}
		On the other hand:
		\begin{equation}
			\widehat{F}_l\left(p\right)= \widehat{\mathcal{L}_l\left\{f\right\}}(p)=F_{0}(p)-F_{1}(p)e_1-F_{2}(p)e_2+F_{3}(p)e_{12},
		\end{equation}
		which clearly are the same expression, thus concluding the proof.
	\end{proof}

	Now, the equivalent result for the reverse involution.
	
	\begin{theo}
		Let $f:\mathbb{R}\to\mathcal{G}_2$ be a component-wise function such that $F_l(p)=\mathcal{L}_l\left\{f\right\}(p)$ exists. Then:
		\begin{equation}
			\mathcal{L}_l\left\{\widetilde{f}\right\}(p)=\widetilde{F}_l\left(p\right),
		\end{equation}
		where $\widetilde{f}$ denotes the reverse involution of $f$, as defined for general multivectors in \eqref{reversemultivector}.
	\end{theo}
	
	\begin{proof}
		As in the previous proof, if $f$ can be expressed as:
		\begin{equation}
			f(t)=f_0(t)+f_1(t)e_1+f_2(t)e_2+f_3(t)e_{12}
		\end{equation}
		for some real-valued functions $f_0,f_1,f_2$ and $f_3$ of a real variable, then:
		\begin{equation}
			\widetilde{f}(t)=f_0(t)+f_1(t)e_1+f_2(t)e_2-f_3(t)e_{12}.
		\end{equation}
		By the linearity property of the geometric Laplace transform, we then have:
		\begin{equation}
			\mathcal{L}_l\left\{\widetilde{f}\right\}(p)=F_{0}(p)+F_{1}(p)e_1+F_{2}(p)e_2-F_{3}(p)e_{12}.
		\end{equation}
		Similarly:
		\begin{equation}
			\widetilde{F}_l\left(p\right) = \widetilde{\mathcal{L}_l\left\{f\right\}}(p) = F_{0}(p)+F_{1}(p)e_1+F_{2}(p)e_2-F_{3}(p)e_{12}.
		\end{equation}
		Thus, the two expressions are identical, concluding the proof.
	\end{proof}
		
	Finally, the result for the geometric conjugate involution.
		
	\begin{theo}
		Let $f:\mathbb{R}\to\mathcal{G}_2$ be a component-wise function such that $F_l(p)=\mathcal{L}_l\left\{f\right\}(p)$ exists. Then:
		\begin{equation}
			\mathcal{L}_l\left\{\overline{f}\right\}(p)=\overline{F}_l\left(p\right),
		\end{equation}
		where $\overline{f}$ denotes the reverse involution of $f$, as defined for general multivectors in \eqref{conjugatemultivector}.
	\end{theo}
	
	\begin{proof}
		Similar to the previous two proofs, this proof follows from a straightforward computation. In particular, if $f$ can be expressed as:
		\begin{equation}
			f(t)=f_0(t)+f_1(t)e_1+f_2(t)e_2+f_3(t)e_{12}
		\end{equation}
		for some functions $f_0,f_1,f_2,f_3;\mathbb{R}\to\mathbb{R}$, then:
		\begin{equation}
			\overline{f}(t)=f_0(t)-f_1(t)e_1-f_2(t)e_2-f_3(t)e_{12}.
		\end{equation}
		By the linearity property of the geometric Laplace transform, we then have:
		\begin{equation}
			\mathcal{L}_l\left\{\overline{f}\right\}(p)=F_{0}(p)-F_{1}(p)e_1-F_{2}(p)e_2-F_{3}(p)e_{12}
		\end{equation}
		and
		\begin{equation}
			\overline{F}_l\left(p\right)=\overline{\mathcal{L}_l\left\{f\right\}}(p)=F_{0}(p)-F_{1}(p)e_1-F_{2}(p)e_2-F_{3}(p)e_{12},
		\end{equation}
		which concludes the proof.
	\end{proof}

	\subsubsection{Geometric Laplace transform of the dual}
	
	As stated in Section \ref{Sec:IntroGA}, the dual operator is one of the key operators in GA. Therefore, it is natural to ask how the geometric Laplace transform behaves with respect to the dual of a component-wise function $f$. It turns out that, similar to the behavior with the three involutions, the geometric Laplace transform of the dual of $f$, denoted $f^\ast$, is simply the dual of the geometric Laplace transform of $f$.
	
	\begin{theo}
		Let $f:\mathbb{R}\to\mathcal{G}_2$ be a component-wise function such that $F_l(p)=\mathcal{L}_l\left\{f\right\}(p)$ exists. Then:
		\begin{equation}
			\mathcal{L}_l\left\{f^\ast\right\}(p)=F_l^\ast\left(p\right),
		\end{equation}
		where $f^\ast$ denotes the dual of $f$, as defined in \eqref{dual}.
	\end{theo}

	\begin{proof}
		Since $f$ is a component-wise function, there exist real-valued functions $f_0,f_1,f_2$ and $f_3$ of a real variable such that:
		\begin{equation}
			f(t) = f_0(t)+f_1(t)e_1+f_2(t)e_2+f_3(t)e_{12}.
		\end{equation}
		The dual of $f$, by the linearity of the dual operator, is the dual of each of its $k$-vector components, i.e., the scalar part, the vector part, and the bivector part. Since $f(t)\in\mathcal{G}_2$ for all $t\in\mathbb{R}$, the dual operator is applied by post-multiplying by $e_{12}$, the pseudoscalar of $\mathcal{G}_2$. Thus, we have:
		\begin{equation}
			\begin{split}
				f(t)^\ast &= \left(f_0(t)\right)^\ast+\left(f_1(t)e_1+f_2(t)e_2\right)^\ast+\left(f_3(t)e_{12}\right)^\ast\\
				&= f_0(t)e_{12} +f_1(t)e_2-f_2(t)e_1-f_3(t),
			\end{split}
		\end{equation}
		where, as noted in Section \ref{Sec:IntroGA}, the following relations hold:
		\begin{equation}
			\begin{split}
				e_1e_{12} &= e_1e_1e_2=e_2,\\
				e_2e_{12} &= -e_2e_{21}=-e_2e_2e_1=-e_1,\\
				e_{12}e_{12} &= -e_{12}e_{21} = -e_1e_2e_2e_1 = -1
			\end{split}
		\end{equation}
		
		By the linearity of the geometric Laplace transform, we obtain:
		\begin{equation}
			\mathcal{L}_l\left\{f\right\}(p) = F_0(t) +F_1(t)e_1+F_2(t)e_2+F_3(t)e_{12}\label{to_transform}
		\end{equation}
		and similarly for the dual
		\begin{equation}
			\mathcal{L}_l\left\{f^\ast\right\}(p) = F_0(t)e_{12} +F_1(t)e_2-F_2(t)e_1-F_3(t).
		\end{equation}
		Finally, taking the dual of identity \eqref{to_transform} yields:
		\begin{equation}
			F_l^\ast(p)=\left(\mathcal{L}_l\left\{f\right\}\right)^\ast(p) = F_0(t)e_{12} +F_1(t)e_2-F_2(t)e_1-F_3(t).
		\end{equation}
		Since the last two expressions are identical, the proof is complete.
	\end{proof}
	
	\subsubsection{Geometric Laplace transform and the push-trough property}
	
	This property, along with the next one, serves as a form of commutativity for geometric Laplace transforms. In particular, they allow rational fractions of geometric Laplace transforms to be expressed as a sum of simpler rational fractions of geometric Laplace transforms, analogous to partial fraction decomposition. This decomposition, in turn, facilitates the identification of the component-wise functions whose geometric Laplace transforms correspond to the original expression. An illustrative example will be presented at the end of the next subsection.
	\begin{theo}\label{push-through}
		Let $f,g:\mathbb{R} \to \mathcal{G}_2$ be two component-wise functions such that $F_l(p)=\mathcal{L}_l\left\{f\right\}(p)$ and $G_l(p)=\mathcal{L}_l\left\{g\right\}(p)$ exist. Then, there exists $\lambda \in \mathbb{R}\setminus\{0\}$ such that:
		\begin{equation}
			F_l(p)G_l(p)=\lambda\left(\overline{F}_l(p)\right)^{-1}G_l(p),
		\end{equation}
		for all values of the multivector variable $p$ for which $F_l(p)$ has a multivector inverse.
	\end{theo}
	
	\begin{rem}\label{multivector_inverses}
		Using the definition of multivector inverse from \cite{Hitzer_inv}, an arbitrary multivector $A$ has a multivector inverse if $A\overline{A}\in\mathbb{R}\setminus\{0\}$. An immediate consequence of this is that, since the geometric conjugate is involutive, the condition $A\overline{A}\in\mathbb{R}\setminus\{0\}$ also ensures that $\overline{A}$ has a multivector inverse. Therefore, $A$ has a multivector inverse if, and only if, $\overline{A}$ has a multivector inverse.
	\end{rem}
	
	\begin{proof}
		It is evident that, using remark \ref{multivector_inverses}, for all values of the multivector variable $p$ for which $F_l(p)$ has a multivector inverse, the following identity holds:
		\begin{equation}
			F_l(p)G_l(p) = F_l(p)\overline{F}_l(p)\left(\overline{F}_l(p)\right)^{-1}G_l(p),
		\end{equation}
		where $F_l(p)\overline{F}_l(p)\in\mathbb{R}$ by definition (since $F_l(p)$ has a multivector inverse). The result is obtained by setting $\lambda = F_l(p)\overline{F}_l(p)$.
	\end{proof}
	
	\begin{rem}
		Theorem \ref{push-through} asserts that, under the previously stated conditions:  
		\begin{equation}
			F_l(p)G_l(p) = \left(\overline{F}_l(p)\right)^{-1}G_l(p)F_l(p)\overline{F}_l(p).
		\end{equation}
		This demonstrates that Theorem \ref{push-through} establishes a form of commutativity for geometric Laplace transforms.
	\end{rem}
	
	\begin{rem}
		The same result holds for the algebras $\mathcal{G}_{1,1}$ and $\mathcal{G}_{0,2}$, as the condition for an arbitrary multivector $A$ in these algebras to have a multivector inverse is that $A\overline{A}\in\mathbb{R}\setminus\{0\}$. In addition, theorem \ref{push-through} can be extended to any geometry algebra $\mathcal{G}_{q,r}$ such that $3\leq p+q\leq 5$. For these algebras, we simply apply the conditions for invertibility provided by \cite{Hitzer_inv}. For instance, in the case of $\mathcal{G}_3$, an arbitrary multivector $A$ has a multivector inverse if $A\overline{A}\widehat{A}\widetilde{A} \in \mathbb{R} \setminus {0}$. Thus, for all values of the multivector variable $p$ for which $F_l(p) \in \mathcal{G}_3$ has a multivector inverse, there exists $\lambda \in \mathbb{R} \setminus {0}$ such that:
		\begin{equation}
			F_l(p)G_l(p)=\lambda\left(\overline{F}_l(p)\widehat{F}_l(p)\widetilde{F}_l(p)\right)^{-1}G_l(p),
		\end{equation}
		where $\lambda=F_l(p)\overline{F}_l(p)\widehat{F}_l(p)\widetilde{F}_l(p)$.
		
		Notice that we also have a similar relationship as the one noted in remark \ref{multivector_inverses}: $F_l(p)$ has a multivector inverse if, and only if, $\overline{F}_l(p)\widehat{F}_l(p)\widetilde{F}_l(p)$ has a multivector inverse. The proof is slightly longer but follows the same essential steps, reducing to showing that:
		\begin{equation}
			\begin{split} &\left(\overline{F}_l(p)\widehat{F}_l(p)\widetilde{F}_l(p)\right)\overline{\left(\overline{F}_l(p)\widehat{F}_l(p)\widetilde{F}_l(p)\right)}\left(\overline{F}_l(p)\widehat{F}_l(p)\widetilde{F}_l(p)\right)^{\widehat{\empty}}\left(\overline{F}_l(p)\widehat{F}_l(p)\widetilde{F}_l(p)\right)^{\widetilde{\empty}}\\
				&=\left(F_l(p)\overline{F}_l(p)\widehat{F}_l(p)\widetilde{F}_l(p)\right)^3,\nonumber
			\end{split}
		\end{equation} 
		which demonstrates that if one side is non-zero, the other is as well. Therefore, the existence of $\left(\overline{F}_l(p)\widehat{F}_l(p)\widetilde{F}_l(p)\right)^{-1}$ is guaranteed by the existence of $F_l(p)^{-1}$, which holds by hypothesis. 
	\end{rem}
	
	\subsubsection{Geometric Laplace transform and the commutative property}
	
	The commutativity property introduced in this section does not hold for all component-wise functions, but only for a specific subset. Notice that, for every multivector $a\in\mathcal{G}_2$, it holds that $a\overline{a}=\overline{a}a\in\mathbb{R}$, where $\overline{a}$ is just the geometric conjugate of $a$ (as stated in Section \ref{Sec:IntroGA}). Therefore, the following relation is satisfied:
	\begin{equation}
		e^{at}e^{\overline{a}t} = e^{\overline{a}t}e^{at}.
	\end{equation}
	Before presenting the main result of this section, we first establish the following technical lemma.
	\begin{lem}\label{commutativity}
		Let $f,g:\mathbb{R}\to\mathcal{G}_2$ be two component-wise functions such that $F_l(p) = \mathcal{L}_l\{f\}(p)$ and $G_l(p) = \mathcal{L}_l\{g\}(p)$ exist. If $f(t)g(t) = g(t)f(t)$ for all $t\in\mathbb{R}$, i.e., they commute, then their geometric Laplace transforms $F_l(p)$ and $G_l(p)$ also commute.
	\end{lem}
	\begin{proof}
		The proof follows from a direct computation. Indeed:
		\begin{align*}
			F_l(p)G_(p) &= \mathcal{L}_l\left\{(f\star g)(t)\right\}(p) &&\text{ by the convolution property (Section \ref{Sec:Conv})}\\
			&= \mathcal{L}_l\left\{(g\star f)(t) \right\}(p) &&\text{ since $f$ and $g$ commute, so do their convolution product}\\
			&= G_l(p)F_l(p)
		\end{align*}
	\end{proof}
	Functions $e^{at}$ and $e^{\overline{a}t}$ satisfy the hypothesis of lemma \ref{commutativity}, and therefore, their geometric Laplace transforms commute. In fact, component-wise exponential functions are the key ingredient of the main result of this section.
	
	However, to proceed, we first need a closed-form expression for the geometric Laplace transform of $e^{at}$ with $a = a_0+a_1e_1+a_2e_2+a_{12}e_12\in\mathcal{G}_2$, which is:
	\begin{equation}
		\mathcal{L}\{e^{at}\}(p) = \left(p-\overline{a}\right)\left(p^2-2a_0p+a\overline{a}\right)^{-1}.
	\end{equation}
	
	This closed-form expression is derived in the second part of Section \ref{ex:exponential}. Notice that, as stated in remark \ref{rem_exp_not_classic}, $p^2-2a_0p+a\overline{a}$ cannot be factored further over $\mathcal{G}_2$, so: 
	\begin{equation}
		(p-\overline{a})\left(p^2-2a_0p+a\overline{a}\right)^{-1}
	\end{equation}
	cannot be simplified further. However, for a subset of multivectors of $\mathcal{G}_2$, such a factorization is possible.
	
	
	\begin{pro}\label{factorizacion_exponencial_proposicion}
		For all multivectors $p$ such that $\langle p\rangle_0 = \langle a\rangle_0$, the expression $p^2-2a_0p+a\overline{a}$ can be factorized as $(p-a)(p-\overline{a})$.
	\end{pro}
	\begin{proof}
		For an arbitrary multivector $p = p_0 + \bm{p}_v = p_0 + p_1 e_1 + p_2 e_2 + p_{12} e_{12} \in \mathcal{G}_2$, if $\langle p\rangle_0 = \langle a\rangle_0$, i.e., if $p_0 = a_0$, then $p^2 - 2a_0 p + a\overline{a}$ is a scalar quantity, that is:
		\begin{equation}
			p^2 - 2a_0 p + a\overline{a} \in \mathbb{R}.
		\end{equation}
		
		Since $a\overline{a} \in \mathbb{R}$, it suffices to verify that $p^2 - 2a_0 p = p^2 - 2p_0 p \in \mathbb{R}$. Expanding this expression yields:
		\begin{equation}
			\begin{split}
				p^2 - 2p_0 p &= (p_0 + \bm{p}_v)(p_0 + \bm{p}_v) - 2p_0(p_0 + \bm{p}_v)\\
				&= p_0^2 + 2p_0 \bm{p}_v + \bm{p}_v^2 - 2p_0^2 - 2p_0 \bm{p}_v\\
				&= -p_0^2 + p_1^2 + p_2^2 - p_{12}^2 \in \mathbb{R}.
			\end{split}
		\end{equation}
		
		Now, let $e^{at}$ and $e^{\overline{a}t}$ be two exponential functions with $a,\overline{a}\in\mathcal{G}_2$ where $\overline{a}$ is the geometric conjugate of $a$ and let us denote by $F_l(p)=$ and $\underline{F}_l(p)$ their geometric Laplace transforms, respectively. Then, by lemma \ref{commutativity}:
			\begin{equation}\label{exp_commute}
			\underline{F}_l(p)F_l(p) = F_l(p)\underline{F}_l(p)
		\end{equation}
		or, equivalently:
		\begin{equation}\label{exp_commute_exp}
			\left(p-a\right)P^{-1}\left(p-\overline{a}\right)P^{-1} = \left(p-\overline{a}\right)P^{-1}\left(p-a\right)P^{-1},
		\end{equation}
		where $P = p^2-2a_0p+a\overline{a}$.
		
		Since $P$ is a scalar quantity, its inverse $P^{-1}$ is also scalar, and therefore it commutes with every multivector. This implies that equation \ref{exp_commute_exp} can be simplified as follows:
		\begin{equation}
			P^{-1}\left(p-a\right)\left(p-\overline{a}\right)P^{-1} = P^{-1}\left(p-\overline{a}\right)\left(p-a\right)P^{-1}
		\end{equation}
		which further reduces to:
		\begin{equation}
			\left(p-a\right)\left(p-\overline{a}\right) = \left(p-\overline{a}\right)\left(p-a\right).
		\end{equation}
		
		Expanding both sides yields:
		\begin{equation}
			p^2 - p\overline{a} - ap + a\overline{a} = p^2 - pa - \overline{a}p + \overline{a}a,
		\end{equation}
		which simplifies to:
		\begin{equation}\label{extended_simp}
			p\overline{a} + ap = pa + \overline{a}p.
		\end{equation}
		
		Recall that, if $a\in\mathcal{G}_2$ can be expressed as $a = a_0 + a_1e_1 + a_2e_2 + a_{12}e_{12}$, then $\overline{a} = a_0 - a_1e_1 - a_2e_2 - a_{12}e_{12}$. Substituting these into \eqref{extended_simp} gives:
		\begin{equation}
			\begin{split}
				&2a_{0}p - p\left(a_{1}e_{1} + a_{2}e_{2} + a_{12}e_{12}\right) + \left(a_{1}e_{1} + a_{2}e_{2} + a_{12}e_{12}\right)p\\
				=&2a_{0}p + p\left(a_{1}e_{1} + a_{2}e_{2} + a_{12}e_{12}\right) - \left(a_{1}e_{1} + a_{2}e_{2} + a_{12}e_{12}\right)p.
			\end{split}
		\end{equation} 
		
		From this equality it follows that:
		\begin{equation}
			p\left(a_{1}e_{1} + a_{2}e_{2} + a_{12}e_{12}\right) = \left(a_{1}e_{1} + a_{2}e_{2} + a_{12}e_{12}\right)p,
		\end{equation}
		which shows that $p$ commutes with the non-scalar part of multivector $a$.
		
		Finally:
		\begin{equation}\label{equation_p2-2a0p+aabar}
			\begin{split}
				(p-a)(p-\overline{a}) &= p^2 - pa - \overline{a}p + \overline{a}a\\
				&= p^2 - ap - \overline{a}p + \overline{a}a\\
				&= p^2 -(a+\overline{a})p+\overline{a}a\\
				&= p^2 -2a_0p+\overline{a}a,
			\end{split}
		\end{equation}
		which completes the proof.
	\end{proof}
	
	\begin{cor}\label{factorizacion_exponencial_corolario}
		For all multivectors $p$ such that $\langle p\rangle_0 = \langle a\rangle_0$, the geometric Laplace transform of the exponential function $e^{at}$ can be simplify to:
		\begin{equation}
			\mathcal{L}\{e^{at}\}(p) = \left(p-a\right)^{-1},
		\end{equation}
		which coincides with its classical counterpart.
	\end{cor}
	\begin{proof}
		The result follows from a straightforward computation. Indeed:
		\begin{align*}
			\mathcal{L}\{e^{at}\}(p) &= \left(p-\overline{a}\right)\left(p^2-2a_0p+a\overline{a}\right)^{-1}\\
			&= \left(p-\overline{a}\right)\left((p-a)(p-\overline{a})\right)^{-1} &&\text{by proposition \ref{factorizacion_exponencial_proposicion}}\\
			&= \left(p-\overline{a}\right)\left(p-\overline{a}\right)^{-1}\left(p-a\right)^{-1}\\
			&= \left(p-a\right)^{-1}.
		\end{align*}
	\end{proof}
	
	\begin{theo}\label{commutativity_theo}
		Let $h:\mathbb{R}\to\mathcal{G}_2$ be a component-wise function such that $H_l(p) = \mathcal{L}_l\{h\}(p)$ exists. If $f(t) = e^{at}$ for a multivector $a\in\mathcal{G}_2$ and $F_l(p) = \mathcal{L}_l\{f\}(p)$ exists, then:
		\begin{equation}
			F_l(p)H_l(p) = \mathcal{L}_l\{e^{\overline{a}t}\}^{-1}(p)H_l(p)F_l(p)\mathcal{L}_l\{e^{\overline{a}t}\}(p)
		\end{equation}
		for all multivectors $p$ such that $\langle p\rangle_0 = \langle a\rangle_0$.
		
		In other words, if $\underline{F}_l(p)$ denotes the geometric Laplace transform of $e^{\overline{a}t}$, then:
		\begin{equation}
			F_l(p)H_l(p) = \underline{F}_l(p)^{-1}(p)H_l(p)F_l(p)\underline{F}_l(p).
		\end{equation}
		or all multivectors $p$ such that $\langle p\rangle_0 = \langle a\rangle_0$.
	\end{theo}
	
	\begin{proof}
		Clearly:
		\begin{equation}\label{eq:base3}
			F_l(p)H_l(p) = \underline{F}_l(p)^{-1}(p)\underline{F}_l(p)F_l(p)H_l(p).
		\end{equation}
		Now, for multivectors $p$ such that $\langle p\rangle_0 = \langle a\rangle_0$, we have, by corollary \ref{factorizacion_exponencial_corolario}:
		\begin{equation}
			\begin{split}
				F_l(p)=\mathcal{L}_l\{e^{at}\} &= (p-a)^{-1},\\
				\underline{F}_l(p)=\mathcal{L}_l\{e^{\overline{a}t}\} &= (p-\overline{a})^{-1}.
			\end{split}
		\end{equation}
		Furthermore, from the proof of proposition \ref{factorizacion_exponencial_proposicion}, we also know that $p$ commutes with the non-scalar part of the multivector $a$, which implies equation \ref{equation_p2-2a0p+aabar}. Applying this relation here, we obtain:
		\begin{align*}
			\underline{F}_l(p)F_l(p) &= (p-\overline{a})^{-1}(p-a)^{-1}\\
			&= \left((p-a)(p-\overline{a})\right)^{-1} && \text{by equation \ref{equation_p2-2a0p+aabar}}\\
			&= \left(p^2-2a_0p+a\overline{a}\right)^{-1}.\\
		\end{align*}
		This expression is, in fact, the geometric Laplace transform of:
		\begin{align*}
				r(t) &= \dfrac{1}{\sqrt{a\overline{a}-a_0^2}}e^{a_0t}\sin\left(\sqrt{a\overline{a}-a_0^2}t\right) && \text{if $a\overline{a}-a_0^2\geq 0$}\\
				r(t) &= \dfrac{1}{\sqrt{a_0^2-a\overline{a}}}e^{a_0t}\sinh\left(\sqrt{a_0^2-a\overline{a}}t\right) && \text{if $a\overline{a}-a_0^2< 0$}
			\end{align*}
		Indeed, if $a\overline{a}-a_0^2\geq 0$, then as shown in Section \ref{ex:sin_sinh}:
		\begin{equation}
				\mathcal{L}\left\{\sin\left(\sqrt{a\overline{a}-a_0^2}t\right)\right\}(p) = \sqrt{a\overline{a}-a_0^2}\big(p^2+a\overline{a}-a_0^2\big)^{-1},
			\end{equation}
		and using Theorem \ref{frecuency_shift}, i.e., the frequency shift property, we obtain:
		\begin{equation}
				\begin{split}
						\mathcal{L}\{r\}(p) &= \dfrac{1}{\sqrt{a\overline{a}-a_0^2}}\sqrt{a\overline{a}-a_0^2}\big((p-a_0)^2+a\overline{a}-a_0^2\big)^{-1}\\
						&= \big(p^2+a_0^2-2a_0p+a\overline{a}-a_0^2\big)^{-1}\\
						&= \big(p^2-2a_0p+a\overline{a}\big)^{-1}
					\end{split}
			\end{equation}
		Analogously, if $a\overline{a}-a_0^2< 0$, then, as shown in Section \ref{ex:sin_sinh}, we have:
		\begin{equation}
				\mathcal{L}\left\{\sinh\left(\sqrt{a_0^2-a\overline{a}}t\right)\right\}(p) = \sqrt{a_0^2-a\overline{a}}\big(p^2-a_0^2+a\overline{a}\big)^{-1},
			\end{equation}
		so:
		\begin{equation}
				\begin{split}
						\mathcal{L}\{r\}(p) &= \dfrac{1}{\sqrt{a_0^2-a\overline{a}}}\sqrt{a_0^2-a\overline{a}}\big((p-a_0)^2-a_0^2+a\overline{a}\big)^{-1}\\
						&= \big(p^2+a_0^2-2a_0p-a_0^2+a\overline{a}\big)^{-1}\\
						&= \big(p^2-2a_0p+a\overline{a}\big)^{-1}
					\end{split}
			\end{equation}
		
		Now, since in both cases $r$ is a real-valued function, it commutes with the multivector-valued, component-wise function $h$. Therefore, by Lemma \ref{commutativity}, we obtain:
		\begin{equation}\label{eq:base2_final}
			\begin{split}
				\underline{F}_l(p)F_l(p)H(p) &= (p^2 - 2a_0p + a\overline{a})^{-1}H(p)\\
				&= H(p)(p^2 - 2a_0p + a\overline{a})^{-1}\\
				&= H(p)\underline{F}_l(p)F_l(p)\\
				&= H(p)F_l(p)\underline{F}_l(p),
			\end{split}
		\end{equation}
		where the last equality follows from the commutativity between $\underline{F}_l(p)$ and $F_l(p)$.
		
		Substituting this result into equation \eqref{eq:base3}, we obtain:
		\begin{equation}
			F_l(p)H(p) = \underline{F}_l(p)^{-1}H(p)F_l(p)\underline{F}_l(p),
		\end{equation}
		which completes the proof.
	\end{proof}
	
	Theorem \ref{commutativity_theo} establishes a form of commutativity for products of geometric Laplace transforms in which at least one factor is a component-wise exponential function. In particular, this theorem allows for the decomposition of products of rational fractions of geometric Laplace transforms, analogous to partial fraction decomposition. For instance, if we set $g(t) = e^{bt}$ and $h(t) = e^{ct}$, where $b, c \in \mathcal{G}_2$ are two distinct multivectors satisfying that $\langle a\rangle_0=\langle b\rangle_0=\langle c\rangle_0$, then the product of the geometric Laplace transforms of $g(t)$, $f(t)$, and $h(t)$ can be expressed as:  
	\begin{equation}\label{ex_decomposition} 
		R(p) = G_l(p)F_l(p)H_l(p) = (p - b)^{-1}(p - a)^{-1}(p - c)^{-1},
	\end{equation}  
	for all multivectors $p$ such that $\langle p\rangle_0 = \langle a\rangle_0=\langle b\rangle_0=\langle c\rangle_0$, which represents a product of three rational functions of geometric Laplace transforms.  
	
	Now, to achieve a decomposition similar to partial fractions, we aim to represent equation \eqref{ex_decomposition} as:
	\begin{equation} 
		R(p) = A_1(p-b)^{-1} + A_2(p-a)^{-1} + A_3(p-c)^{-1}, 
	\end{equation} 
	where $A_1$, $A_2$, and $A_3$ are three distinct multivectors determined as follows: 
	\begin{equation} 
		\begin{split} 
			A_1 &= R(p)(p-b)\big|_{p=b} = (p-b)^{-1}(p-a)^{-1}(p-c)^{-1}(p-b)\big|_{p=b},\\
			A_2 &= R(p)(p-a)\big|_{p=a} = (p-b)^{-1}(p-a)^{-1}(p-c)^{-1}(p-a)\big|_{p=a},\\
			A_3 &= R(p)(p-c)\big|_{p=c} = (p-b)^{-1}(p-a)^{-1}\big|_{p=c}.
		\end{split}
	\end{equation}
	
	The computation of $A_3$ is straightforward: 
	\begin{equation} 
		A_3 = (c-b)^{-1}(c-a)^{-1}. 
	\end{equation} 
	For $A_1$ and $A_2$, theorem \ref{commutativity_theo} allows us to write: 
	\begin{equation} 
		\begin{split} 
			A_1 &= (p-b)^{-1}(p-a)^{-1}(p-c)^{-1}(p-b)\\
			&= (p-\overline{b})(p-a)^{-1}(p-\overline{b})^{-1}(p-b)^{-1}(p-c)^{-1}(p-b)\\
			&= (p-\overline{b})(p-a)^{-1}(p-\overline{b})^{-1}(p-\overline{b})(p-c)^{-1}(p-\overline{b})^{-1}(p-b)^{-1}(p-b)\\
			&= (p-\overline{b})(p-a)^{-1}(p-c)^{-1}(p-\overline{b})^{-1}\\
			A_2 &= (p-b)^{-1}(p-a)^{-1}(p-c)^{-1}(p-a)\\
			&= (p-b)^{-1}(p-\overline{a})(p-c)^{-1}(p-\overline{a})^{-1}(p-a)^{-1}(p-a)\\
			&= (p-b)^{-1}(p-\overline{a})(p-c)^{-1}(p-\overline{a})^{-1}, 
		\end{split} 
	\end{equation} 
	which, when evaluated at $p=b$ and $p=a$, simplify to: 
	\begin{equation} 
		\begin{split} 
			A_1 &= (b-\overline{b})(b-a)^{-1}(b-c)^{-1}(b-\overline{b})^{-1}\\
			A_2 &= (a-b)^{-1}(a-\overline{a})(a-c)^{-1}(a-\overline{a})^{-1}.\\
		\end{split} 
	\end{equation} 
	
	\begin{rem}
		Theorem \ref{commutativity_theo} holds for the algebras $\mathcal{G}_{1,1}$ and $\mathcal{G}_{0,2}$ because, in these algebras, $a$ commutes with its geometric conjugate $\overline{a}$. Therefore, $e^{at}$ also commutes with $e^{\overline{a}t}$.  
		
		For geometric algebras $\mathcal{G}_{q,r}$ with $q + r = 3$, an analogous version of theorem \ref{commutativity_theo} can be established. In this case, the commutativity involves $a$ and the product $\overline{a} \hat{a} \tilde{a}$, where $\hat{a}$ and $\tilde{a}$ represent the main involution and reverse involution of $a$, respectively. It is known that these elements commute \citep{Hitzer_inv}. Furthermore, the expressions derived in the same paper, which yield a scalar when multiplied by the multivector $a$, can be applied to extend theorem \ref{commutativity_theo} to geometric algebras $\mathcal{G}_{q,r}$ with $3<q+r\leq 5$.
	\end{rem}
	\color{black}
	
	\section{Geometric Laplace transform of well-known functions}\label{Sec:Ex}
	
	This section is dedicated to presenting the geometric Laplace transform for the extension of certain well-known functions to multivector-valued, component-wise functions of a real variable. These expressions are formulated in $\mathcal{G}_2$, although analogous expressions can be derived for higher-dimensional geometric algebras.
	
	\subsection{Dirac delta function}\label{ex:diract}
	
	The standard Dirac delta function is defined as:
	\begin{equation}
		\delta(t-c) =  \left\{\begin{split}
			0\hspace{0.5cm}\text{if $t\neq c$}\\
			\infty\hspace{0.5cm}\text{if $t=c$}
		\end{split}\right.,
	\end{equation}
	for $c\in\mathbb{R}$. In addition, it satisfies the property:
	\begin{equation}
		\int_{-\infty}^{\infty}\delta(t-c)dt = 1,
	\end{equation}
	which reduces to: 
	\begin{equation}
		\int_{0}^{\infty}\delta(t-c)dt = 1,
	\end{equation}
	if $c\geq0$.
	
	Its extension to a component-wise function is given by:
	\begin{equation}
		\delta_g(t-c) = \delta(t-c)K,
	\end{equation}
	where $K\in\mathcal{G}_2$ is a multivector.
	
	Now, we first compute the geometric Laplace transform of the standard Dirac delta function for $c\geq0$:
	\begin{align}
		\mathcal{L}\{\delta(t-c)\}(p) &= \int_{0}^{\infty}e^{-pt}\delta(t-c)dt\nonumber\\
		&=e^{-pc}\int_{0}^{\infty}\delta(t-c)dt && \text{since }\delta(t-c)\neq 0\text{ at }t=c\label{delta_laplace}\\
		&= e^{-pc}. && \text{main property of $\delta$}\nonumber
	\end{align}
	Therefore, the geometric Laplace transform of $\delta_g$ is:
	\begin{align}
		\mathcal{L}_l\left\{\delta_g\right\}(p) &= \int_0^{\infty} e^{-pt}\delta(t-c)Kdt\nonumber\\
		&=\left(\int_0^{\infty} e^{-pt}\delta(t-c)dt\right)K && K\text{ does not depend on }t\\
		&=e^{cp}K, && \text{using equation \eqref{delta_laplace}}\nonumber
	\end{align}
	where $K\in\mathcal{G}_2$ is a multivector.
	
	\subsection{Heaviside or unit step function}\label{ex:step}
	
	The standard unit step function is defined as:
	\begin{equation}
		u(t) =  \left\{\begin{split}
			1\hspace{0.5cm}\text{if $x\geq 0$}\\
			0\hspace{0.5cm}\text{if $x<0$}
		\end{split}\right.,
	\end{equation}
	
	For a given multivector $K\in\mathcal{G}_2$, its extension to a component-wise function is given by:
	\begin{equation}
		u_g(t) = u(t)K.
	\end{equation}
	
	Now, as in the case of the Dirac delta function, we start by computing the geometric Laplace transform of the standard unit step function:
	\begin{align}
		\mathcal{L}\{u\}(p) &= \int_{0}^{\infty}e^{-pt}u(t)dt\nonumber\\
		&=\int_{0}^{\infty}e^{-pt}dt\label{unit_step_laplace}\\
		&= -p^{-1}e^{-pt}\big|_0^{\infty} = p^{-1}.\nonumber
	\end{align}
	where, since $|u(t)|\leq 1$, the function $u$ is of exponential order 0 and, therefore, as stated in theorem \ref{existence_theorem}, the integral converges for $\langle p\rangle_0>0$.
	
	In turn, this allows us to compute the geometric Laplace transform of $u_g$:
	\begin{align}
		\mathcal{L}_l\left\{u_g\right\}(p) &= \int_0^{\infty} e^{-pt}u(t)Kdt\nonumber\\
		&=\left(\int_0^{\infty} e^{-pt}u(t)dt\right)K && K\text{ does not depend on }t\\
		&=p^{-1}K, && \text{using equation \eqref{unit_step_laplace}}\nonumber
	\end{align}
	where, as before, $K\in\mathcal{G}_2$ is a multivector, and the integral converges for $\langle p\rangle_0>0$.
	
	\subsection{Ramp function}\label{ex:ramp}
	
	The standard ramp function is defined as:
	\begin{equation}
		r(t) =  \left\{\begin{split}
			x\hspace{0.5cm}\text{if $x\geq 0$}\\
			0\hspace{0.5cm}\text{if $x<0$}
		\end{split}\right.,
	\end{equation}
	
	Its extension to a component-wise function is given by:
	\begin{equation}
		r_g(t) = r(t)K,
	\end{equation}
	where $K\in\mathcal{G}_2$ is a multivector.
	
	Since:
	\begin{equation}
		r(t) = \int_{-\infty}^t u(s)ds = \int_{0}^t u(s)ds,
	\end{equation}
	and:
	\begin{equation}
		\mathcal{L}\{u\}(p) = p^{-1},
	\end{equation}
	we can compute the geometric Laplace transform of $r$ using theorem \ref{integral_property_theorem} from Section \ref{integral_property}:
	\begin{align}	
		\mathcal{L}_l\left\{r\right\}(p) &= \mathcal{L}_l\left\{\int_{0}^t u(s)ds\right\}(p)\nonumber\\
		&= p^{-1}\mathcal{L}_l\left\{u\right\}(p)\nonumber\\
		&= p^{-1}p^{-1}\label{ramp_laplace}\\
		&= p^{-2},\nonumber
	\end{align}
	where as before, the integral converges for $\langle p\rangle_0 > 0$.
	
	\begin{rem}
		The same strategy, i.e., using theorem \ref{integral_property_theorem} from Section \ref{integral_property}, can be used to easily compute the geometric Laplace transform of the positive branch of the parabola function, $r_2(t)=x^2$, the positive branch of the cubic function, $r_3(t)=x^3$, and so on. 
	\end{rem}
	
	Now, as in previous examples, for a multivector $K\in\mathcal{G}_2$, the geometric Laplace transform of $r_g$ can be computed as follows:
	\begin{align}
		\mathcal{L}_l\left\{r_g\right\}(p) &= \int_0^{\infty} e^{-pt}r(t)Kdt\nonumber\\
		&=\left(\int_0^{\infty} e^{-pt}r(t)dt\right)K && K\text{ does not depend on }t\\
		&=p^{-2}K, && \text{using equation \eqref{ramp_laplace}}\nonumber
	\end{align}
	where, as before, the integral converges for $\langle p\rangle_0>0$.
	 
	 \subsection{Sine and hyperbolic sine functions}\label{ex:sin_sinh}
	 
	 The geometric Laplace transform of the standard sine and hyperbolic sine functions is computed in a similar manner. Therefore, we will first proceed with a step-by-step computation of the geometric Laplace transform of the standard sine function. For $\omega\in\mathbb{R}$, the geometric Laplace transform is given by:
	 \begin{equation}\label{integral_sine}
	 	\mathcal{L}_l\left\{\sin\right\}(p) = \int_0^\infty e^{-pt}\sin(\omega t)dt.
	 \end{equation}
	 Using integration by parts with:
	 \begin{equation}
	 	\left\{\begin{matrix}
	 		u = e^{-pt}\hspace{1.2cm} & \rightarrow & du = -pe^{-pt}dt \\
	 		v = -\frac{1}{\omega}\cos(\omega t) & \leftarrow & dv = \sin(\omega t)dt
	 	\end{matrix}\right\}
	 \end{equation}
	 we have:
	 \begin{equation}
	 	\int_0^\infty e^{-pt}\sin(\omega t)dt = -\left.e^{-pt}\frac{1}{\omega}\cos(\omega t)\right|_0^\infty - \int_0^\infty pe^{-pt}\frac{1}{\omega}\cos(\omega t)dt,
	 \end{equation}
	 where, since $|\sin(\omega t)|\leq 1$, the sine function is of exponential order 0. Therefore, as stated in theorem \ref{existence_theorem}, the integral converges for $\langle p\rangle_0>0$, and the following identity also holds:
	 \begin{equation}\label{boundary}
	 	\lim\limits_{t\to\infty}e^{-pt}\cos(\omega t) = 0.
	 \end{equation}
	 
	 Now, applying integration by parts again:
	 \begin{equation}
	 	\left\{\begin{matrix}
	 		u = pe^{-pt}\hspace{0.8cm} & \rightarrow & du = -p^2e^{-pt}dt \\
	 		v = \frac{1}{\omega^2}\sin(\omega t) & \leftarrow & \hspace{0.2cm}dv = \frac{1}{\omega}\cos(\omega t)dt
	 	\end{matrix}\right\}
	 \end{equation}
	 we obtain:
	 {\small\begin{equation}
	 	\begin{split}
	 		&\int_0^\infty e^{-pt}\sin(\omega t)dt\\
	 		&= -\left.e^{-pt}\frac{1}{\omega}\cos(\omega t)\right|_0^\infty - \left[\left.pe^{-pt}\frac{1}{\omega^2}\sin(\omega t)\right|_0^\infty + \int_0^\infty p^2e^{-pt}\frac{1}{\omega^2}\sin(\omega t)dt\right]\\
	 		&= -\left.e^{-pt}\frac{1}{\omega}\cos(\omega t)\right|_0^\infty - \left.pe^{-pt}\frac{1}{\omega^2}\sin(\omega t)\right|_0^\infty - \frac{p^2}{\omega^2}\int_0^\infty e^{-pt}\sin(\omega t)dt,\\
	 	\end{split}
	 \end{equation}}
 	 where, once again, the integral converges for $\langle p\rangle_0>0$. In addition, similarly to equation \eqref{boundary}, we have that:
 	 \begin{equation}\label{boundary1}
 	 	\lim\limits_{t\to\infty}e^{-pt}\sin(\omega t) = 0.
 	 \end{equation}
	 
	 The integral on the right is the same as the integral on the left, which allows us to rewrite the expression as:
	 {\small\begin{equation}
	 	\left(1 + \frac{p^2}{\omega^2}\right)\int_0^\infty e^{-pt}\sin(\omega t)dt = -\left.e^{-pt}\frac{1}{\omega}\cos(\omega t)\right|_0^\infty - \left.pe^{-pt}\frac{1}{\omega^2}\sin(\omega t)\right|_0^\infty
	 \end{equation}}	 
	 
	 Now, the boundary terms converge as $t$ tends to infinite if $\langle p\rangle_0>0$. If so, we can use equations \eqref{boundary} and \eqref{boundary1} to evaluate them, obtaining:
	 \begin{equation}
	 	\left(\frac{\omega^2 + p^2}{\omega^2}\right)\int_0^\infty e^{-pt}\sin(\omega t)dt = \frac{1}{\omega},
	 \end{equation}
	 which gives us:
	 \begin{equation}\label{sine_laplace}
	 	\int_0^\infty e^{-pt}\sin(\omega t)dt = \omega \left(p^2 + \omega^2\right)^{-1}.
	 \end{equation}
	 
	 Regarding the standard hyperbolic sine function, its geometric Laplace transform is defined as:
	 \begin{equation}\label{integral_sinh}
	 	\mathcal{L}_l\{\sinh\}(p) = \int_{0}^{\infty}e^{-pt}\sinh(\omega t)dt,
	 \end{equation}
	 where $\omega\in\mathbb{R}$.
	 
	 The integral in \eqref{integral_sinh} can be solved using integration by parts twice, similar to the computation for the sine function. Since the process is entirely analogous, we omit the detailed steps here. The only distinction lies in the derivative of the hyperbolic cosine function, given by $\frac{d}{dt}\cosh(\omega t) = \omega\sinh(\omega t)$, as opposed to the trigonometric counterpart $\frac{d}{dt}\cos(\omega t) = -\omega\sin(\omega t)$. As a consequence, the geometric Laplace transform of the standard hyperbolic sine function is:
	 \begin{equation}\label{sinh_laplace}
	 	\mathcal{L}_l\{\sinh\}(p) = \omega\left(p^2-\omega^2\right)^{-1}.
	 \end{equation}
	 
	 Now, the component-wise extension of the standard sine and hyperbolic sine functions is defined as:  
	 \begin{equation}
	 	\begin{split}
	 		\sin_g(\omega t) = \sin(\omega t)K,\\
	 		\sinh_g(\omega t) = \sinh(\omega t)K,
	 	\end{split}
	 \end{equation}  
	 where $K$ is a multivector of $\mathcal{G}_2$, and $\omega\in\mathbb{R}$.  
	 
	 Finally, the geometric Laplace transform of $\sin_g$ and $\sinh_g$ are given by:  
	 \begin{align}
	 	\mathcal{L}_l\left\{\sin_g\right\}(p) &= \int_0^{\infty}e^{-pt}\sin(\omega t)Kdt\nonumber\\
	 	&= \left(\int_0^{\infty}e^{-pt}\sin(\omega t)dt\right)K && K\text{ does not depend on }t\\
	 	&= \omega\left(p^2 + \omega^2\right)^{-1}K && \text{using equation \eqref{sine_laplace}}\nonumber
	 \end{align}
	 and
	 \begin{align}
	 	\mathcal{L}_l\left\{\sinh_g\right\}(p) &= \int_0^{\infty}e^{-pt}\sinh(\omega t)Kdt\nonumber\\
	 	&= \left(\int_0^{\infty}e^{-pt}\sinh(\omega t)dt\right)K && K\text{ does not depend on }t\\
	 	&= \omega\left(p^2 - \omega^2\right)^{-1}K, && \text{using equation \eqref{sinh_laplace}}\nonumber
	 \end{align}
	 where both integrals converge for $\langle p\rangle_0>0$.
	 
	 \begin{rem}\label{rem_different_sines}
	 	Let $f$ be the following component-wise function:  
	 	\begin{equation}
	 		f(t) = A_0g_0(t) + A_1g_1(t)e_1 + A_2g_2(t)e_2 + A_3g_3(t)e_{12},
	 	\end{equation}
	 	where $g_i$ can either be $g_i= \sin(\omega_i t)$ or $g_i = \sinh(\omega_i t)$ for$i=0,1,2,3$, and where $A_i$ ($i=0,1,2,3$) denote different amplitudes, and $\omega_i$ ($i=0,1,2,3$) denote different frequencies. Then, its geometric Laplace transform is given by:
	 	\begin{equation}
	 		\begin{split}
	 			\mathcal{L}_l\left\{f\right\}(p) &= \sum_{i=0}^{3}A_i\mathcal{L}_l\{g_i\}(p)\\
	 			&=\left\{\begin{split}
	 				&\sum_{i=0}^{3}A_i\omega_i\left(p^2 +\omega_i^2\right)^{-1}e_i\hspace{1cm}\text{if }g_i=\sin(\omega_i t)\\
	 				\empty\\
	 				&\sum_{i=0}^{3}A_i\omega_i\left(p^2 -\omega_i^2\right)^{-1}e_i\hspace{1cm}\text{if }g_i=\sinh(\omega_i t)
	 				\end{split}\right.
	 		\end{split}\nonumber
	 	\end{equation}
	 	where $e_0$ and $e_{3}$ denote the scalar and pseudoscalar basis elements of $\mathcal{G}_2$, i.e., $1$ and $e_{12}$.
	 	
	 	This result follows directly from the linearity of the geometric Laplace transform and equations \eqref{sine_laplace} and \eqref{sinh_laplace}.
	 \end{rem}
	 
	 \subsection{Cosine and hyperbolic cosine functions}\label{ex:cos_cosh}
	 
	 As in the previous case, the geometric Laplace transform of the cosine and hyperbolic cosine functions are computed in a similar way. Therefore, we begin by computing the geometric Laplace transform of the standard cosine function. To do so, we utilize the relationship:
	 \begin{equation}
	 	\frac{d}{dt}\sin(\omega t) = \omega\cos(\omega t),
	 \end{equation}
	 and apply the geometric Laplace transform to both sides:
	 \begin{equation}
	 	\mathcal{L}_l\left\{\frac{d}{dt}\sin\right\}(p) = \omega\mathcal{L}_l\left\{\cos\right\}(p)
	 \end{equation}
	 Now, using the property of the derivative introduced in Section \ref{derivative1}, we obtain:
	 \begin{equation}
	 	\begin{split}
	 		\omega\mathcal{L}_l\left\{\cos\right\} &= p\mathcal{L}_l\{sin\}(p)- \sin(\omega0)\\
	 		&=p\omega\left(p^2 + \omega^2\right)^{-1} - \sin(0\omega)\\
	 		&=p\omega\left(p^2 + \omega^2\right)^{-1}.
	 	\end{split}
	 \end{equation}
	 Thus, the geometric Laplace transform of the standard cosine function is:
	 \begin{equation}\label{cosine_laplace}
	 	\mathcal{L}_l\left\{\cos\right\}(p) = p\left(p^2 + \omega^2\right)^{-1},
	 \end{equation}
	 where the integral converges for $\langle p\rangle_0>0$ since the cosine function is also of exponential order 0 ($|\cos(\omega t)|\leq1$).
	 
	 Analogously, using the relation:
	 \begin{equation}
	 	\frac{d}{dt}\sinh(\omega t) = \omega\cosh(\omega t),
	 \end{equation}
	 we obtain:
	 \begin{equation}\label{cosh_laplace}
	 	\mathcal{L}_l\{\cosh\}(p) = p\left(p^2-\omega^2\right)^{-1},
	 \end{equation}
	 where, once again, we have used the derivative property of the geometric Laplace transform, as established in Section \ref{derivative1}.
	 
	 The component-wise extension of the standard cosine and hyperbolic cosine functions is defined as:
	 \begin{equation}
	 	\begin{split}
	 		\cos_g(\omega t) = \cos(\omega t)K,\\
	 		\cosh_g(\omega t) = \cosh(\omega t)K,
	 	\end{split}
	 \end{equation}  
	 where $K$ is a multivector of $\mathcal{G}_2$, $\omega\in\mathbb{R}$, and their geometric Laplace transforms are computed as:
	 \begin{align}
	 	\mathcal{L}_l\left\{\cos_g\right\}(p) &= \int_0^{\infty}e^{-pt}\cos(\omega t)Kdt\nonumber\\
	 	&= \left(\int_0^{\infty}e^{-pt}\cos(\omega t)dt\right)K && K\text{ does not depend on }t\\
	 	&= p\left(p^2 + \omega^2\right)^{-1}K && \text{using equation \eqref{cosine_laplace}}\nonumber
	 \end{align}
	 and
	 \begin{align}
	 	\mathcal{L}_l\left\{\cosh_g\right\}(p) &= \int_0^{\infty}e^{-pt}\cosh(\omega t)Kdt\nonumber\\
	 	&= \left(\int_0^{\infty}e^{-pt}\cosh(\omega t)dt\right)K && K\text{ does not depend on }t\\
	 	&= p\left(p^2 - \omega^2\right)^{-1}K, && \text{using equation \eqref{cosh_laplace}}\nonumber
	 \end{align}
	 where, again, both integrals converge for $\langle p\rangle_0>0$.
	 
	 \begin{rem}
	 	The geometric Laplace transform of a sum of cosine or hyperbolic cosine functions with different amplitudes and frequencies is computed analogously to the process described in remark \ref{rem_different_sines}, i.e., it is the sum of each amplitude multiplied by the geometric Laplace transform of the corresponding cosine or hyperbolic cosine function.
	 \end{rem}
	 
	 \subsection{Exponential function}\label{ex:exponential}
	 
	 For $a\in\mathbb{R}$, the geometric Laplace transform of the standard exponential function is defined as:
	 \begin{equation}\label{integral_exp}
	 	\mathcal{L}_l\left\{e\right\}(p) = \int_0^\infty e^{-pt}e^{at}dt.
	 \end{equation}
	 Since $a$ and $p$ always commute, we can rewrite the integral in equation \eqref{integral_exp} as:
	 \begin{equation}\label{integral_exp2}
	 	\int_0^\infty e^{-pt}e^{at}dt = \int_0^\infty e^{-(p-a)t}dt.
	 \end{equation}
	Now, since $|e^{at}| < e^{2|a|t}$, where $|a|$ denotes the absolute value of $a$, it follows that $e^{at}$ is of exponential order $2|a|$. Therefore, according to theorem \ref{existence_theorem}, the integrals in equations \eqref{integral_exp} and \eqref{integral_exp2} converge for $\langle p\rangle_0 > a$. Hence:  
	\begin{equation}\label{exp_laplace}
		\int_0^\infty e^{-(p-a)t}dt = \left.-(p-a)^{-1}e^{-(p-a)t}\right|_0^\infty = (p-a)^{-1},
	\end{equation}
	where the last equality follows from the fact that, for $\langle p\rangle_0 > a$, the following also holds:  
	\begin{equation}
		\lim\limits_{t\to\infty} e^{-(p-a)t} = 0.
	\end{equation}
	Now, if we define $e_g(at) = e^{at}K$ for a multivector $K\in\mathcal{G}_2$ and $a\in\mathbb{R}$, then:
	\begin{align}
		\mathcal{L}_l\left\{e_g\right\}(p) &= \int_0^{\infty} e^{-pt}e^{at}Kdt\nonumber\\
		&=\left(\int_0^{\infty} e^{-pt}e^{at}dt\right)K && K\text{ does not depend on }t\\
		&=(p-a)^{-1}K, && \text{using equation \eqref{exp_laplace}}\nonumber
	\end{align}
	which is the geometric Laplace transform of $e_g(at) = e^{at}K$.
	
	Finally, if $a\in\mathcal{G}_2$, then $a$ and $p$ do not need to commute, but we can compute the geometric Laplace transform of $e^{at}$ as follows. Using equation \eqref{exp_expanded}, we have that:
	\begin{equation}
		e^{at} = \left\{\begin{split}
			e^{a_0t}\left(\cos(a_v't) + \frac{\bm{a}_v}{a_v'}\sin(a_v't)\right)&\qquad\text{if }a_v^2<0\\
			\empty\\
				e^{a_0t}\left(\cosh(a_vt) + \frac{\bm{a}_v}{a_v}\sinh(a_vt)\right)&\qquad\text{if }a_v^2>0
		\end{split}\right.
	\end{equation}
	where $a = a_0+a_1e_1+a_2e_2+a_{12}e_{12}$, $\bm{a}_v = a_1e_1+a_2e_2+a_{12}e_{12}$, $a_v = \sqrt{a_1^2+a_2^2-a_{12}^2}$, and $a_v'=\sqrt{a_{12}^2-a_1^2-a_2^2}$.
	
	Now, using the linearity of the geometric Laplace transform, we have that:
	\begin{equation}
		\mathcal{L}_l\left\{e^{at}\right\}(p) = \left\{\begin{split}
			&\mathcal{L}_l\left\{e^{a_0t}\cos(a_v't)\right\}(p) + \frac{\bm{a}_v}{a_v'}\mathcal{L}_l\left\{e^{a_0t}\sin(a_v't)\right\}(p)\\
			&\mathcal{L}_l\left\{e^{a_0t}\cosh(a_vt)\right\}(p) + \frac{\bm{a}_v}{a_v}\mathcal{L}_l\left\{e^{a_0t}\sinh(a_vt)\right\}(p)
		\end{split}\right.
	\end{equation}
	which, using the frequency-shift property introduced in theorem \ref{freq_shift} and the geometric Laplace transforms of the trigonometric and hyperbolic functions, turns into:
	\begin{equation}
		\begin{split}
			\mathcal{L}_l\left\{e^{at}\right\}(p) &= \left\{\begin{split}
				&\left(p-a_0\right)\left((p-a_0)^2+a_v'^2\right)^{-1} + \frac{\bm{a}_v}{a_v'}a_v'\left((p-a_0)^2+a_v'^2\right)^{-1}\\[1.5mm]
				&\left(p-a_0\right)\left((p-a_0)^2-a_v^2\right)^{-1} + \frac{\bm{a}_v}{a_v}a_v\left((p-a_0)^2-a_v^2\right)^{-1}
			\end{split}\right.\\[1.5mm]
			&=\left(p-\overline{a}\right)\left(p^2 -2a_0p+a\overline{a}\right)^{-1},
		\end{split}
	\end{equation}
	thus completing the computation.
	
	\begin{rem}\label{rem_exp_not_classic}
		Notice that $\mathcal{L}_l\left\{e^{at}\right\}(p)\neq \left(p-a\right)^{-1}$, which makes it the first geometric Laplace transform that differs from its classical counterpart. This difference arises from the lack of commutativity in geometric algebras larger than $\mathcal{G}_{p,q}$ with $p+q=1$. In these commutative cases, $p^2 - 2a_0p + a\overline{a}$ factorizes into $(p-a)(p-\overline{a})$, which yields the classical result. However, in higher-dimensional algebras this factorization no longer holds, and in general
		\begin{equation}
				p^2 - 2a_0p + a\overline{a} \neq (p-a)(p-\overline{a}).
		\end{equation}
		Indeed,
		\begin{equation}
			\begin{split}
				(p-a)(p-\overline{a}) &= p^2-p\overline{a}-ap+a\overline{a}\\
				&= p^2-(p\overline{a}+ap)+a\overline{a},
			\end{split}			 
		\end{equation}
		where $(p\overline{a}+ap)\neq 2a_0p$. To see this, let $p = p_0+p_1e_1+p_2e_2+p_{12}e_{12} = p_0 + \bm{p}_v$ and $a = a_0+a_1e_1+a_2e_2+a_{12}e_{12} = a_0 + \bm{a}_v$. Then:
		\begin{equation}
			\begin{split}
				p\overline{a}+ap &= p(a_0-\bm{a}_v)+(a_0 +\bm{a}_v)p\\
				&= pa_0-p\bm{a}_v+a_0p+\bm{a}_vp\\
				&= 2a_0p-p_0a_0-\bm{p}_v\bm{a}_v+a_0p_0+\bm{a}_v\bm{p}_v\\
				&= 2a_0p+\bm{a}_v\bm{p}_v-\bm{p}_v\bm{a}_v,
			\end{split}
		\end{equation}
		with $\bm{a}_v\bm{p}_v-\bm{p}_v\bm{a}_v\neq 0$, since $\bm{a}_v$ and $\bm{p}_v$ do not commute in general.
	\end{rem}

	\section{Conclusions}\label{Sec:Conc}
	
	In the present work, the geometric Laplace transform has been introduced as an extension of the classical Laplace transform for component-wise multivector-valued functions of a real variable. This establishes a formal framework for utilizing geometric algebra in applications where the Laplace transform is required, such as in the works of \cite{Vel23} and \cite{DVZM2024}. In addition, all the properties of the classical Laplace transform have been formulated and proven for the geometric Laplace transform, along with several new properties. Finally, the geometric Laplace transforms of certain functions have been computed to serve as a conversion table, facilitating further computations.
	
	This new extension will enable a broader application of the Laplace transform in fields where geometric algebra has only recently been adopted, such as signal processing and control theory. As a result, it will facilitate more efficient analysis and problem-solving involving multivector-valued functions.
	
	\bibliographystyle{apalike}
	\bibliography{Laplace_bib}
\end{document}